\documentclass[letterpaper, 10 pt, conference]{ieeeconf}

\IEEEoverridecommandlockouts
\usepackage{amsmath, amssymb, amsfonts, mathtools}

\usepackage{amsthm}
\usepackage{bm}
\usepackage{graphicx}
\usepackage{booktabs}
\usepackage{placeins}
\usepackage{hyperref}
\usepackage{cleveref}
\usepackage[dvipsnames]{xcolor}
\let\labelindent\relax
\usepackage{enumitem}
\hypersetup{
  colorlinks=true,
  linkcolor=Blue,
  citecolor=Blue,
  urlcolor=Blue
}
\usepackage[ruled,vlined,linesnumbered]{algorithm2e}
\usepackage{tikz}
\newcommand{\legendbox}[1]{%
\setlength{\fboxsep}{0pt}%
  \fcolorbox{black!40}{#1}{\phantom{\rule{0.45cm}{0.18cm}}}%
}
\definecolor{phaseblue}{HTML}{90CAF9}
\definecolor{phaseorange}{HTML}{FFCC80}
\definecolor{phasepurple}{HTML}{CE93D8}

\crefformat{figure}{Figure~#2#1#3}
\crefformat{table}{Table~#2#1#3}
\crefformat{section}{Section~#2#1#3}
\crefformat{subsection}{Section~#2#1#3}
\crefformat{subsubsection}{Section~#2#1#3}
\crefformat{theorem}{Theorem~#2#1#3}
\crefformat{lemma}{Lemma~#2#1#3}
\crefformat{corollary}{Corollary~#2#1#3}
\crefformat{proposition}{Proposition~#2#1#3}
\crefformat{definition}{Definition~#2#1#3}
\crefformat{assumption}{Assumption~#2#1#3}
\crefformat{remark}{Remark~#2#1#3}
\crefformat{example}{Example~#2#1#3}
\crefformat{algorithm}{Algorithm~#2#1#3}
\crefformat{equation}{(#2#1#3)}
\crefformat{appendix}{Appendix~#2#1#3}
\crefformat{problem}{Problem~#2#1#3}
\crefmultiformat{figure}{Figures~#2#1#3}{ and~#2#1#3}{, #2#1#3}{, and~#2#1#3}
\crefmultiformat{table}{Tables~#2#1#3}{ and~#2#1#3}{, #2#1#3}{, and~#2#1#3}
\crefmultiformat{section}{Sections~#2#1#3}{ and~#2#1#3}{, #2#1#3}{, and~#2#1#3}
\crefmultiformat{subsection}{Sections~#2#1#3}{ and~#2#1#3}{, #2#1#3}{, and~#2#1#3}
\crefmultiformat{subsubsection}{Sections~#2#1#3}{ and~#2#1#3}{, #2#1#3}{, and~#2#1#3}
\crefmultiformat{theorem}{Theorems~#2#1#3}{ and~#2#1#3}{, #2#1#3}{, and~#2#1#3}
\crefmultiformat{lemma}{Lemmas~#2#1#3}{ and~#2#1#3}{, #2#1#3}{, and~#2#1#3}
\crefmultiformat{corollary}{Corollaries~#2#1#3}{ and~#2#1#3}{, #2#1#3}{, and~#2#1#3}
\crefmultiformat{proposition}{Propositions~#2#1#3}{ and~#2#1#3}{, #2#1#3}{, and~#2#1#3}
\crefmultiformat{definition}{Definitions~#2#1#3}{ and~#2#1#3}{, #2#1#3}{, and~#2#1#3}
\crefmultiformat{assumption}{Assumptions~#2#1#3}{ and~#2#1#3}{, #2#1#3}{, and~#2#1#3}
\crefmultiformat{remark}{Remarks~#2#1#3}{ and~#2#1#3}{, #2#1#3}{, and~#2#1#3}
\crefmultiformat{example}{Examples~#2#1#3}{ and~#2#1#3}{, #2#1#3}{, and~#2#1#3}
\crefmultiformat{algorithm}{Algorithms~#2#1#3}{ and~#2#1#3}{, #2#1#3}{, and~#2#1#3}
\crefmultiformat{equation}{(#2#1#3)}{ and~(#2#1#3)}{, (#2#1#3)}{, and~(#2#1#3)}
\crefmultiformat{appendix}{Appendices~#2#1#3}{ and~#2#1#3}{, #2#1#3}{, and~#2#1#3}
\crefmultiformat{problem}{Problems~#2#1#3}{ and~#2#1#3}{, #2#1#3}{, and~#2#1#3}

\crefrangeformat{figure}{Figures~#3#1#4--#5#2#6}
\crefrangeformat{table}{Tables~#3#1#4--#5#2#6}
\crefrangeformat{section}{Sections~#3#1#4--#5#2#6}
\crefrangeformat{subsection}{Sections~#3#1#4--#5#2#6}
\crefrangeformat{subsubsection}{Sections~#3#1#4--#5#2#6}
\crefrangeformat{theorem}{Theorems~#3#1#4--#5#2#6}
\crefrangeformat{lemma}{Lemmas~#3#1#4--#5#2#6}
\crefrangeformat{corollary}{Corollaries~#3#1#4--#5#2#6}
\crefrangeformat{proposition}{Propositions~#3#1#4--#5#2#6}
\crefrangeformat{definition}{Definitions~#3#1#4--#5#2#6}
\crefrangeformat{assumption}{Assumptions~#3#1#4--#5#2#6}
\crefrangeformat{remark}{Remarks~#3#1#4--#5#2#6}
\crefrangeformat{example}{Examples~#3#1#4--#5#2#6}
\crefrangeformat{algorithm}{Algorithms~#3#1#4--#5#2#6}
\crefrangeformat{equation}{(#3#1#4)--(#5#2#6)}
\crefrangeformat{appendix}{Appendices~#3#1#4--#5#2#6}
\crefrangeformat{problem}{Problems~#3#1#4--#5#2#6}

\makeatletter
\def\endtable{\end@float}
\def\endfigure{\end@float}
\makeatother
\usepackage[caption=false]{subfig}
\usepackage{mathrsfs}
\usepackage[mode=buildnew]{standalone}
\newtheoremstyle{ieeeblock}
{6pt}{6pt}{\normalfont}{}{\bfseries}{.}{0.5em}{}
\theoremstyle{ieeeblock}
\newtheorem{theorem}{Theorem}

\newtheorem{assumption}{Assumption}
\newtheorem{definition}{Definition}
\newtheorem{remark}{Remark}
\newtheorem{problem}{Problem}
\newcommand{\E}{\mathbb{E}}
\newcommand{\Prob}{\mathrm{Pr}}
\newcommand{\CVaR}{\mathrm{CVaR}}
\newcommand{\VaR}{\mathrm{VaR}}
\newcommand{\R}{\mathbb{R}}
\newcommand{\ccmppi}{{\scshape{CcMppi}}}

\usepackage{parskip}

\title{\LARGE \bf Risk-Constrained Belief-Space Optimization\\for Safe Control under Latent Uncertainty}

\author{Clinton Enwerem\textsuperscript{1}, John S. Baras\textsuperscript{1}, Calin Belta\textsuperscript{1}%
\thanks{\textsuperscript{1}The authors are with the Institute for Systems Research, University of Maryland, College Park, MD, USA. Emails:~\texttt{\char`\{enwerem, baras, calin\char`\}@umd.edu}.}
\thanks{\copyright~2026 IEEE. Personal use of this material is permitted. Permission from IEEE must be obtained for all other uses, in any current or future media, including reprinting/republishing this material for advertising or promotional purposes, creating new collective works, for resale or redistribution to servers or lists, or reuse of any copyrighted component of this work in other works. Accepted for publication at the 65th IEEE Conference on Decision and Control (CDC).}
}

\begin{document}

\maketitle
\thispagestyle{empty}
\pagestyle{empty}

\begin{abstract}
Many safety-critical control systems operate under latent uncertainty that sensors cannot resolve at decision time. Such uncertainty, arising from unknown physical properties, disturbances, or unobserved geometry, affects dynamics, task feasibility, and safety margins. Standard methods optimize expected performance and offer limited protection against rare but severe outcomes, while robust formulations treat uncertainty conservatively without exploiting its probabilistic structure. We consider systems with measured state and an unknown, time-invariant parameter represented by a belief distribution. We propose a risk-sensitive belief-space Model Predictive Path Integral (MPPI) controller that plans under this belief, regularizes performance using Conditional Value-at-Risk (CVaR), and imposes a CVaR constraint on a trajectory safety margin over the horizon. For the exact risk-constrained formulation underlying this controller, we establish three properties: (1) the CVaR constraint implies a probabilistic safety guarantee, (2) the controller recovers the risk-neutral optimum as the objective risk weight tends to zero, and (3) a union-bound argument extends the per-horizon guarantee to cumulative safety over repeated solves. In contact-rich MuJoCo simulations of vision-guided dexterous stowing, where a manipulator inserts a grasped object into an occupied slot with pose uncertainty exceeding prescribed lateral clearance requirements, our method achieves 82\% success with zero contact violations at high risk aversion, compared with 55\% and 50\% for a risk-neutral configuration and a chance-constrained baseline, both of which incur nonzero exterior contact forces. Project page: \url{https://clintonenwerem.com/belief-cvar-mppi/}.
\end{abstract}

\section{Introduction}\label{sec:intro}
To illustrate the main problem addressed in this paper,  consider a manipulator that must place a grasped object into a nonempty receptacle (see \Cref{fig:runex} and \cite{park2025vulcan}). Successful execution requires estimating receptacle pose, transporting the object safely, and placing it with sufficient clearance while avoiding excessive contact forces. Even with perfect proprioception, noisy visual estimates of receptacle pose and geometry introduce latent uncertainty that, if unmodeled, degrades performance and causes constraint violation. Safe control therefore requires belief representations that support online updates and enable risk-aware actions. Beyond manipulation, safe control under latent uncertainty arises in autonomous driving in mixed traffic~\cite{tariqAutonomousVehicleOvertaking2022}, multi-robot navigation with partial observations~\cite{vandenberg2012motion}, and quadrotor attitude control under unknown wind disturbances~\cite{vahsBeliefControlBarrier2023a}.
\setlength{\fboxsep}{0pt}
\setlength{\fboxrule}{1pt}
\begin{figure}
\centering
\newlength{\figH}
\setlength{\figH}{0.18\textheight}
\begin{minipage}[t]{0.42\linewidth}
    \centering
    \fbox{\includegraphics[height=\figH,width=\dimexpr\linewidth-2\fboxrule\relax,keepaspectratio]{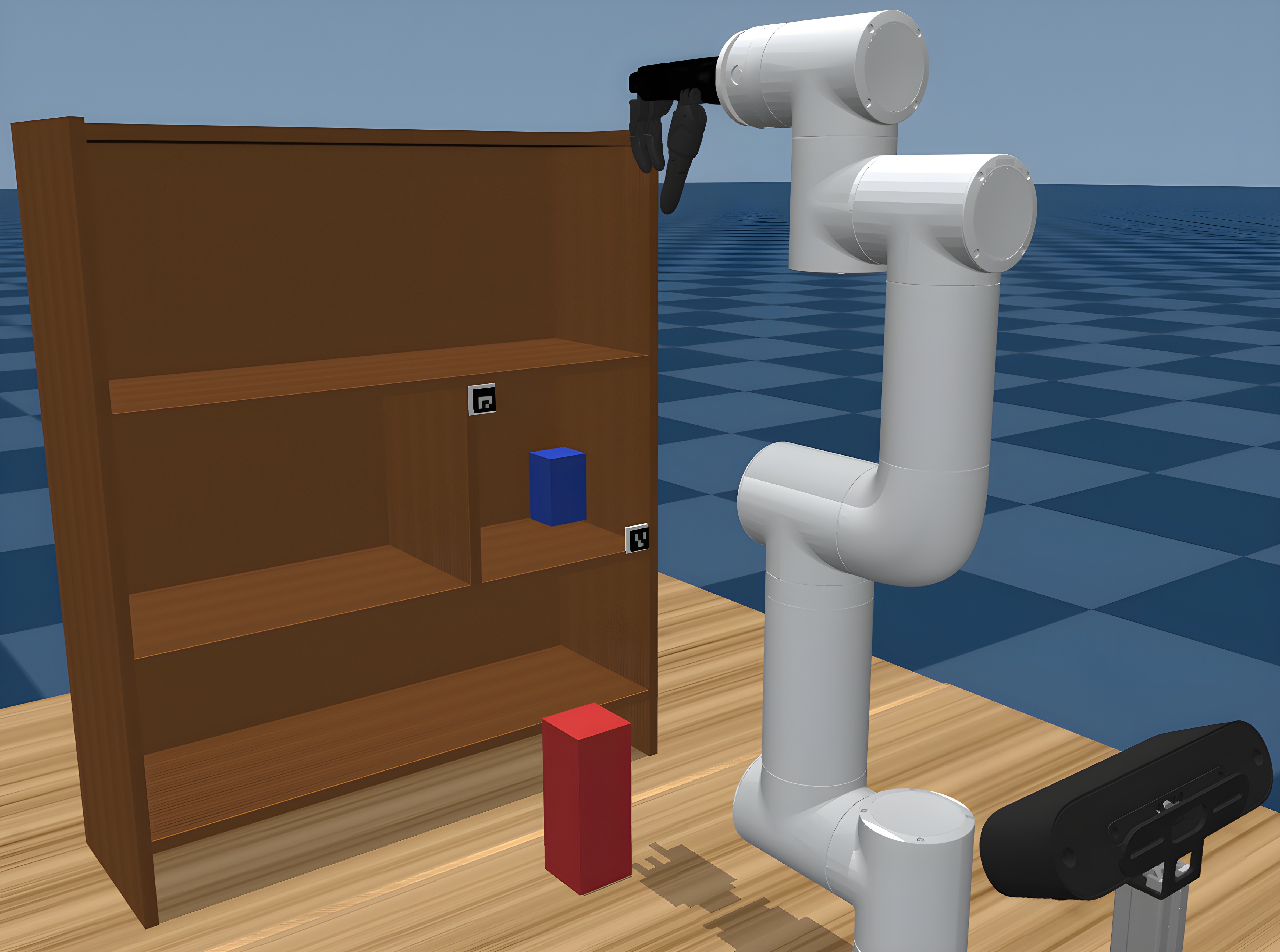}}
    \\[-2pt]
    {\small (a) Object Stowing Task.}
\end{minipage}\hfill%
\begin{minipage}[t]{0.27\linewidth}
    \centering
    \fbox{\begin{tikzpicture}
        \node[inner sep=0pt] (img) {\includegraphics[height=\figH, width=\dimexpr\linewidth-2\fboxrule\relax, keepaspectratio]{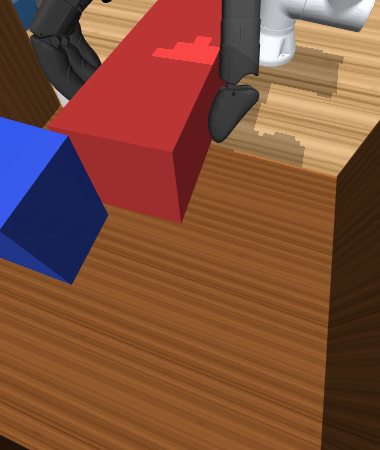}};
        \node[anchor=south east,xshift=-2pt,yshift=3pt,rounded corners=2pt,
              fill=red!15,draw=red!60!black,text=red!70!black,font=\scriptsize]
              at (img.south east) {{\sf Collision}};
    \end{tikzpicture}}\\[-2pt]
    {\small (b) Risk-Neutral.}
\end{minipage}\hfill%
\begin{minipage}[t]{0.27\linewidth}
    \centering
    \fbox{\begin{tikzpicture}
        \node[inner sep=0pt] (img) {\includegraphics[height=\figH, width=\dimexpr\linewidth-2\fboxrule\relax, keepaspectratio]{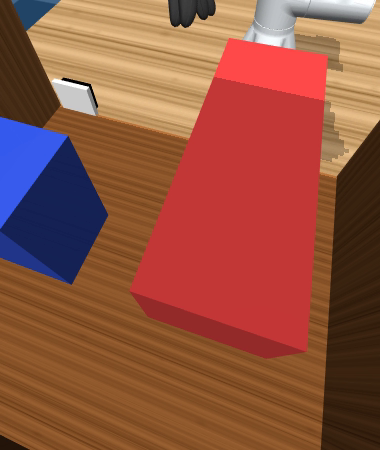}};
        \node[anchor=south east,xshift=-2pt,yshift=3pt,rounded corners=2pt,
              fill=green!15,draw=green!50!black,text=green!40!black,font=\scriptsize]
              at (img.south east) {{\sf Success}};
    \end{tikzpicture}}\\[-2pt]
    {\small (c) Risk-Aware.}
\end{minipage}
\caption{\textbf{Running Example.} (a). A robotic manipulator inserts a grasped object into an occupied receptacle with uncertain pose. (b). With low risk aversion ($\beta_s=0.50$), the controller allows trajectories that produce catastrophic contact forces. (c). With high risk aversion ($\beta_s=0.95$), the CVaR safety constraint maintains safer insertion margins and avoids contact. Here, $\beta_s$ is the CVaR confidence level for the safety constraint. \Cref{fig:stowing_sequences} shows the corresponding success and failure sequences.}
\label{fig:runex}
\end{figure}

\subsection{Related Work}\label{ssec:relwork}
\subsubsection{Risk Measures \& Risk-Sensitive Optimization} Coherent risk measures~\cite{artzner1999coherent}, particularly Conditional Value-at-Risk (CVaR)~\cite{rockafellar2000cvar}, are standard tools for capturing tail risk in stochastic programming~\cite{shapiro2014lectures}. In robotics, they appear in CVaR-constrained Markov Decision Processes (MDPs)~\cite{chow2015risk}, risk-sensitive reinforcement learning~\cite{noorani2025risk}, motion planning~\cite{enwerem2024robust,ahmadi_risk-averse_2021,hakobyan_risk-aware_2019}, and axiomatic risk assessment~\cite{majumdar2020risk}; see~\cite{akellaRiskAwareRoboticsTail2024} for a survey.

\subsubsection{Model-Predictive Control (MPC) \& Chance Constraints} Chance-constrained MPC~\cite{yin2024chance} replaces hard constraints with probabilistic violation bounds~\cite{nemirovski2006convex,mesbah2016stochastic}, admitting exact linear reformulations under Gaussian uncertainty for linear systems~\cite{blackmore2011chance}. In the nonlinear stochastic setting, sampling-based methods such as Model Predictive Path Integral (MPPI)~\cite{williams2017mppi} bypass the need for gradient computation entirely.

\subsubsection{Belief-Space Planning}Belief-space planning~\cite{platt2010belief,kaelbling1998pomdp} provides a principled framework for acting under state and parameter uncertainty. Particle filter representations~\cite{thrun2005probabilistic} are standard, integrating naturally with sampling-based planners~\cite{vandenberg2012motion} and belief-space feedback policies~\cite{agha2014firm}.

\subsubsection{Safety-Critical Control}
Control Barrier Functions (CBFs)~\cite{ames2017cbf} certify forward invariance of safe sets for control-affine systems, with extensions to stochastic~\cite{fushimi2025safe}, higher-order~\cite{xiao_high-order_2022}, and multiagent~\cite{enwerem2024safe,glotfelterNonsmoothBarrierFunctions2017} settings; see~\cite{Xiao2026LearnFeasSafeCtrl,cohen2023adaptive} for broader treatments. Closely related, \cite{vahs2026safety} guarantees probabilistic safety via belief-based control barrier functions, but omits risk measures in the safety constraint and latent parameters in the belief dynamics.

Our work differs in that uncertainty lies in \emph{latent parameters} influencing dynamics, costs, and safety margins, rather than only the state or disturbance distribution. We maintain a posterior over such parameters and invoke CVaR \emph{both} in the objective and in a trajectory safety constraint, a combination that is, to the best of our knowledge, novel.
\smallskip

\noindent\textbf{Contributions \& Outline:}
\vspace{-6pt}
\begin{enumerate}
\item \emph{Belief-Space Uncertainty Representation (\Cref{sec:prbform}):} We model a dynamical system whose dynamics, cost, and safety constraints depend on a latent parameter inferred indirectly from observations.

\item \emph{Risk-Constrained Belief-Space Optimization Formulation (\Cref{sec:method}):} We formulate safe control as a belief-space receding-horizon trajectory optimization problem that trades off expected performance against tail risk while enforcing CVaR safety constraints.
\item \emph{Belief-Space MPPI Control Synthesis Algorithm}: We present a belief-space MPPI algorithm for solving the optimization problem in \Cref{sec:method}.
\item \emph{Theoretical Properties (\Cref{sec:theory}):} We prove that the CVaR safety constraint in the risk-sensitive mathematical program introduced in \Cref{sec:method} implies a probabilistic safety guarantee, that the controller recovers risk-neutral behavior in the appropriate limit of the risk confidence level, and that the per-horizon guarantee extends to cumulative safety over repeated receding-horizon application via a union bound.
\item \emph{Experimental Validation}: We validate our results using physics-based simulations of a dexterous object relocation task within MuJoCo (\Cref{sec:casest}), comparing our results with a point-estimate chance-constrained baseline.
\end{enumerate}

\subsection{Notation and Preliminaries}\label{ssec:prelim}

We write \(\E[\cdot]\) for expectation, \(\Prob(\cdot)\) for probability, and \((z)^+ \text{ for } \max\{z, 0\}\). All random quantities are defined on a common probability space.
We use \(t\) for absolute discrete time and \(k \in \{0,\dots,H\}\) for the relative step index within a prescribed prediction horizon of length \(H\). We will also denote by \(x_{t_1:t_2}\) the finite sequence, given by \((x_{t_1}, x_{t_1+1}, \ldots, x_{t_2})\), where \(t_1\) and \(t_2\) are distinct discrete time steps, and \(t_2\,{>}\,t_1\).

\begin{definition}[Value-at-Risk]\label{def:var}
For a random variable \(Z\) with confidence level \(\beta \in (0,1)\), the Value-at-Risk is
\begin{equation}\label{eq:vardef}
\VaR_\beta(Z) = \inf\bigl\{z \in \R : \Prob(Z \le z) \ge \beta\bigr\}.
\end{equation}
\end{definition}

\begin{definition}[Conditional Value-at-Risk]\label{def:cvar}
For a random variable \(Z\) with finite first moment and confidence level \(\beta \in (0,1)\), we have
\begin{equation}\label{eq:cvardef}
\CVaR_\beta(Z) = \inf_{\eta \in \R}\left\{\eta + \frac{1}{1-\beta}\,\E\!\left[(Z - \eta)^+\right]\right\}.
\end{equation}
Here, \(\CVaR_\beta(Z)\) is the conditional expectation of \(Z\) beyond its \(\beta\)-quantile, satisfying \(\VaR_\beta(Z)\le \CVaR_\beta(Z)\)~\cite{rockafellar2000cvar}. For minimization, \(\beta{=}0.50\) is roughly risk-neutral, while \(\beta\in[0.9,1)\) emphasizes increasingly severe tail outcomes.
\end{definition}

\section{Problem Formulation}\label{sec:prbform}
We consider a stochastic system with measured state and a latent parameter,
\begin{equation}\label{eq:dyn}
x_{t+1} = f(x_t, u_t, \theta, w_t), \qquad z_t = g(x_t, \theta, v_t),
\end{equation}
where \(x_t \in \mathbb{R}^n\), \(u_t \in \mathcal{U} \subseteq \mathbb{R}^m\), \(\theta \in \Theta \subseteq \mathbb{R}^p\) is an unknown time-invariant latent parameter vector, \(z_t \in \mathbb{R}^q\), and \(f\), \(g\) are nonlinear process and measurement functions. At each decision time, \(x_t\) is available to the controller, whereas \(z_t\) provides indirect information about \(\theta\). The measurement model \(g\) induces the likelihood \(p(z_t~{\mid}~\theta, x_t)\), and the noise signals \(w_t \sim p_w\), \(v_t \sim p_v\) are mutually independent with known distributions, \(p_w\) and \(p_v\). Thus, the uncertainty represented by the belief is a fixed but unknown model parameter, rather than an unobserved state, an additive disturbance, or time-varying model error. Over a horizon of \(H\) steps starting from time \(t\), the control sequence \(\mathbf{u} = u_{t:t+H-1}\) lies in the \(H\)-fold Cartesian product \(\mathcal{U}^H \subseteq \mathbb{R}^{mH}\). The trajectory cost is
\vspace{-4pt}
\begin{equation}\label{eq:cost}
J_H(\mathbf{u}, \theta, x_t) = \sum_{k=0}^{H-1} \ell(x_{t+k}, u_{t+k}, \theta) + \phi(x_{t+H}, \theta),
\end{equation}
where \(\ell \colon \mathbb{R}^n \times \mathcal{U} \times \Theta \to \mathbb{R}_{\ge 0}\) is the stage cost and \(\phi \colon \mathbb{R}^n \times \Theta \to \mathbb{R}_{\ge 0}\) is the terminal cost. Since \(\theta\) is unknown and \(w_t\) is random, \(J_H(\mathbf{u}, \theta, x_t)\) is a random variable whose distribution is induced by \(b_t(\theta)\) and the process noise over the horizon. 
\begin{remark}[Shorthand Notation]\label{rem:expdep}
Hereafter, we write \(J_H\) for \(J_H(\mathbf{u}, \theta, x_t)\), with dependence on \(\mathbf{u}\), \(x_t\), and \(\theta\) implicit via~\eqref{eq:dyn}. We use \(J_H(\mathbf{u})\) when explicit control dependence is central, e.g., in \Cref{asm:jhfinmom,asm:uhcomp,thm:limit}, and \(J_H(\cdot)\) for arbitrary control arguments.
\end{remark}
\begin{definition}[Belief and Latent Parameter]\label{def:belief}
Since \(\theta\) is not directly observable, the controller maintains a posterior belief
\begin{equation}\label{eq:belup}
b_t(\theta) \;:=\; p\bigl(\theta \mid x_{0:t},\, z_{1:t}\bigr),
\end{equation}
where \(p(\cdot)\) is the joint probability law over \((\theta, x_{0:t}, z_{1:t})\), and \((x_{0:t}, z_{1:t})\) is the information available up to time \(t\). Accordingly, \(b_t\) is a parameter belief; a partially observed state would instead require a joint state--parameter belief.
\end{definition}
Exact inference over \(b_t(\theta)\) is generally intractable for nonlinear state and measurement models with latent parameter dependence, since the posterior admits no closed form. We therefore use a particle filter, which can represent multimodal beliefs, handle hybrid continuous-discrete latent states, and is simple to implement, requiring only the number of particles \(N_p\) to be specified~\cite{vahs2026safety,thrun2005probabilistic}.
\begin{definition}[Particle Belief]\label{def:partbel}
We approximate the belief \(b_t(\theta)\) by a set of \(N_p\) weighted particles \(\{(\theta^{(i)}, \omega_t^{(i)})\}_{i=1}^{N_p}\), with \(\omega_t^{(i)} \ge 0\) and \(\sum_i \omega_t^{(i)} = 1\). Upon applying control \(u_t\) and receiving observation \(z_{t+1}\), each particle is propagated through the process model~\eqref{eq:dyn},
\begin{equation}
    x_{t+1}^{(i)} = f\bigl(x_t, u_t, \theta^{(i)}, w_t^{(i)}\bigr), \ w_t^{(i)} \sim p_w,
\end{equation}
and the \emph{particle weights} are updated via the likelihood induced by the measurement model~\eqref{eq:dyn}
\begin{align}\label{eq:partbelupd}
\tilde{\omega}_{t+1}^{(i)} &= \omega_t^{(i)}\cdot p\bigl(z_{t+1} \mid \theta^{(i)}, x_{t+1}^{(i)} \bigr), \\
\omega_{t+1}^{(i)}         &= \frac{\tilde{\omega}_{t+1}^{(i)}}{\sum_{j=1}^{N_p} \tilde{\omega}_{t+1}^{(j)}},
\end{align}
with resampling applied when the effective sample size \(\hat{N}_{\mathrm{eff}} = \bigl(\sum_i (\omega_t^{(i)})^2\bigr)^{-1}\) falls below a threshold \(N_{\mathrm{thr}} \in \mathbb{N}\).
\end{definition}
\begin{definition}[Safety Function]\label{def:safety}
We encode safety by a continuously-differentiable function \(h(x,\theta)\) that defines the parameter-dependent safe set \(\mathcal{S}(\theta) = \{x \in \R^n \colon h(x,\theta) \ge 0\}\), with interior \(\mathrm{Int}\,\mathcal{S}(\theta)\), and boundary \(\partial\mathcal{S}(\theta)\). At time, \(t\), the safety function assumes the following values:
\vspace{-4pt}
\begin{equation}\label{eq:safe}
h(x_t,\theta)\begin{cases} > 0, & x_t \in \mathrm{Int}\,\mathcal{S}(\theta),\\ = 0, & x_t \in \partial \mathcal{S}(\theta),\\ < 0, & x_t \notin \mathcal{S}(\theta). \end{cases}
\end{equation}
\end{definition}\vspace{-7pt}\noindent \vspace{0pt}The function \(h\) may represent physical safety margins such as obstacle clearance or contact force limits (see \Cref{sec:casest}). Since instantaneous safety does not capture safety over the task horizon, we next define a trajectory-level safety notion.
\begin{definition}[Trajectory Safety Margin]\label{def:trjsfmargin}
Given a trajectory~$x_{t:t+H}{\equiv}\,\mathbf{x}$ (induced by \(\mathbf{u}\) and \(\theta\) via \eqref{eq:dyn}) with horizon \(H\) starting at time \(t\), its safety margin is\vspace{-4pt}
\begin{equation}\label{eq:margin}
M_H(\mathbf{x}, \theta) = \min_{k=0,\dots,H} \, h(x_{t+k}, \theta).
\end{equation}
\noindent For notational simplicity, hereafter, we will drop the \(\mathbf{x}\) and \(\theta\) arguments in \eqref{eq:margin} and simply write \(M_H\) (the exceptions in \Cref{rem:expdep} carry over). A positive \(M_H\) indicates the trajectory remains in \(\mathcal{S}(\theta)\), while \(M_H < 0\) indicates a violation, hence \({-}{M_H}\) may be interpreted as a \emph{violation variable}. Because \(M_H\) records the worst margin over the horizon, it distinguishes violation depth but not the number or duration of violations attaining the same minimum.
\end{definition}
\begin{problem}[Risk-Constrained Belief-Space Control]\label{prb:rsbel}
Given a stochastic system with latent parameters \(\theta\), 
posterior belief \(b_t(\theta)\) at time \(t\), planning horizon \(H\), stage and terminal costs defining the horizon cost \(J_H\), and trajectory safety margin \(M_H\), find a control sequence \(\mathbf{u}^\star \in \mathcal{U}^H\) such that
\begin{align}
\mathbf{u}^\star \in \arg\min_{\mathbf{u} \in \mathcal{U}^H} \quad
& \E_{b_t}[J_H] + \lambda_r\,\CVaR_{\beta_c}(J_H) \label{eq:obj} \\
\text{subject to} \quad
& \CVaR_{\beta_s}({-}{M_H}) \le 0. \label{eq:cvarcnst}
\end{align}
\end{problem}
\begin{remark}[Underlying Distribution and Problem Data]\label{rem:distprbdata}In \Cref{prb:rsbel}, the expectation and risk measures are taken with respect to the posterior belief \(b_t(\theta)\) and process noise over the planning horizon. The parameter \(\beta_c\) controls tail sensitivity in the performance objective, \(\beta_s\) determines the degree of safety conservativeness and provides a hard probabilistic bound on the average tail safety-margin violation (see \Cref{thm:safety}), and \(\lambda_r \in (0,1)\) weights expected performance against tail risk. Thus, while objective \eqref{eq:obj} favors accurate and efficient task completion, the risk constraint \eqref{eq:cvarcnst} regulates the worst-case tail of the trajectory safety margin.
\end{remark}
\begin{remark}[Monotonicity in \(\beta_c\)]\label{rem:mono}
Since \(\CVaR_{\beta_c}(Z)\) equals the conditional expectation of \(Z\) above its \(\beta_c\)-quantile, it is non-decreasing in \(\beta_c\)~\cite{rockafellar2000cvar}. Consequently, increasing \(\beta_c\) with \(\lambda_r\,{>}\,0\) strictly increases the risk penalty in~\eqref{eq:obj}, producing more conservative trajectories with respect to the performance cost's variability.
\end{remark}

\section{Risk-Sensitive Belief-Space MPPI}\label{sec:method}
\begin{algorithm}[t]
\caption{Risk-Sensitive Belief-Space MPPI}
\label{alg:belmpc}
\DontPrintSemicolon
\KwIn{State \(x_t\), Prior belief \(b_{t-1}(\theta)\)}
\KwResult{Control \(u_t\)}
Observe measurement \(z_t\) and update belief \(b_{t-1}(\theta)\rightarrow b_t(\theta)\) via \eqref{eq:partbelupd}\;
Sample \(\theta^{(i)} \sim b_t(\theta)\) for \(i = 1, \dots, N_p\)\;
\For{\(j = 1, \dots, K\)}{
    Sample control sequence \(\mathbf{u}^{(j)}\)\;
    \For{\(i = 1, \dots, N_p\)}{
        Simulate trajectory via~\eqref{eq:rollout}\;
        Compute cost \(J_H^{(i,j)}\) and margin \(M_H^{(i,j)}\)\;
    }
    Estimate \(\E_{b_t}[J_H^{(j)}]\), \(\CVaR_{\beta_c}(J_H^{(j)})\), \(\CVaR_{\beta_s}(-M_H^{(j)})\)\;
    \tcp*[l]{Update \(\mathbf{u}^{(j)}\) via MPPI weighting~\cite{williams2017mppi}}
    Compute score \(S^{(j)}\) via~\eqref{eq:score}\;
    Compute weight \(\rho^{(j)}\) via~\eqref{eq:mppiweight}\;
}
\(\mathbf{u}^\star \leftarrow \hat{\mathbf{u}} + \sum_{j=1}^{K} \rho^{(j)}\, \boldsymbol{\epsilon}^{(j)}\) via~\eqref{eq:mppiupdate}\;
Apply first element \(u_t^\star\) of \(\mathbf{u}^\star\)\;
\end{algorithm}

\subsection{MPPI Algorithm}\label{ssec:alg}
Following \cite{williams2017mppi}, we solve \Cref{prb:rsbel} via MPPI importance sampling-based trajectory evaluation (Algorithm~\ref{alg:belmpc}). At each control step, we draw \(N_p\) latent parameter samples \(\theta^{(i)} \sim b_t(\theta)\), and roll out \(K\) candidate control sequences through~\eqref{eq:dyn}
\begin{equation}\label{eq:rollout}
\hspace{-5.8pt}x_{t+k+1}^{(i)} = f\bigl(x_{t+k}^{(i)}, u_{t+k}^{(j)}, \theta^{(i)}, w_{t+k}^{(i)}\bigr), \ k = 0, \dots, H-1,
\end{equation}
where \(j =1, \ldots, K\) indexes the candidate sequence and \(i = 1, \dots, N_p\) indexes the belief particle. Each candidate sequence is a perturbation of the current estimate \(\hat{\mathbf{u}} = \hat{u}_{t:t+H-1} \in \mathcal{U}^H\), initialized at zero for the first solve and thereafter by shifting the previous MPPI solution, by i.i.d.\ Gaussian noise (\(\boldsymbol{\epsilon}^{(j)}\))
\begin{equation}\label{eq:ucand}
\mathbf{u}^{(j)} = \hat{\mathbf{u}} + \boldsymbol{\epsilon}^{(j)},\
\boldsymbol{\epsilon}^{(j)} = (\epsilon_{t:t+H-1}^{(j)}),\
\epsilon_{t+k}^{(j)} \sim \mathcal{N}(0,\Sigma),
\end{equation}
so that \(u_{t+k}^{(j)} = \hat{u}_{t+k} + \epsilon_{t+k}^{(j)}\) is the \(k\)-th element of \(\mathbf{u}^{(j)}\), with $\Sigma \in \mathbb{R}^{m\times m}$ denoting the sample covariance. The rollout procedure~\eqref{eq:rollout} produces per-particle trajectory costs \(\{J_H^{(i,j)}\}^{N_p}_{i=1}\) and safety margins \(\{M_H^{(i,j)}\}^{N_p}_{i=1}\) for each candidate \(j\), where \(J_H^{(i,j)}\) and \(M_H^{(i,j)}\) are evaluated on the trajectory \(\mathbf{x}^{(i,j)}\) induced by \(\mathbf{u}^{(j)}\) and particle \(\theta^{(i)}\) via~\eqref{eq:dyn}. From these, we obtain a per-candidate score
\begin{equation}\label{eq:score}
\begin{aligned}
S^{(j)}&=\widehat{\E}[J_H^{(j)}] + \lambda_r\,\widehat{\CVaR}_{\beta_c}\!\left(J_H^{(j)}\right) \\
&+ \mu\,\Bigl(\widehat{\CVaR}_{\beta_s}\!\left(-M_H^{(j)}\right)\Bigr)^+, \ \text{where }
\end{aligned}
\end{equation}\vspace{-5pt}
\begin{equation}\label{eq:empexp}
\begin{aligned}
\widehat{\E}[J_H^{(j)}]
&= \frac{1}{N_p}\sum_{i=1}^{N_p} J_H^{(i,j)}, \\[-0.6ex]
\widehat{\CVaR}_{\beta_c}\!\left(J_H^{(j)}\right)
&= \frac{1}{\lceil(1-\beta_c)N_p\rceil}
   \sum_{i \in \mathcal{I}_{\beta_c}^{(j)}} J_H^{(i,j)},\text{ and}\\[-0.6ex]
\widehat{\CVaR}_{\beta_s}\!\left(-M_H^{(j)}\right)
&= \frac{1}{\lceil(1-\beta_s)N_p\rceil}
   \sum_{i \in \mathcal{I}_{\beta_s}^{(j)}} \bigl(-M_H^{(i,j)}\bigr)
\end{aligned}
\end{equation}
are the empirical expected cost, performance risk, and safety risk corresponding to candidate \(j\), and \(\mu>0\) is the penalty weight for a positive empirical safety risk. In \eqref{eq:empexp}, \(\mathcal{I}_{\beta_c}^{(j)}\) indexes the \(\lceil(1-\beta_c)N_p\rceil\) highest values of \(\{J_H^{(i,j)}\}^{N_p}_{i=1}\), and \(\mathcal{I}_{\beta_s}^{(j)}\) indexes the \(\lceil(1-\beta_s)N_p\rceil\) largest values of \(\{-M_H^{(i,j)}\}^{N_p}_{i=1}\). Each candidate control sequence \(\mathbf{u}^{(j)} \in \mathcal{U}^H\) is assigned an \emph{importance sampling weight} \(\rho^{(j)}\) \cite{williams2017mppi} computed via
\begin{equation}\label{eq:mppiweight}
\frac{1}{\eta}\exp\!\left(-\frac{1}{\lambda}\left(S^{(j)}+\lambda
\sum_{k=0}^{H-1}
\bigl(\hat{u}_{t+k}-\bar{u}_{t+k}\bigr)^\top
\Sigma^{-1}\epsilon_{t+k}^{(j)}\right)\right),
\end{equation}
where $\eta\,{=}\,\sum_j \rho^{(j)}$ normalizes the weights, $\lambda>0$ is the temperature that controls the sharpness of the weight distribution across control sequence candidates, $\hat{u}_{t+k}$ is the current control estimate, $\bar{u}_{t+k}$ is the nominal distribution mean, and $\epsilon_{t+k}^{(j)}$ is the perturbation defined in \eqref{eq:ucand}. The nominal mean is a proportional task-space velocity toward the active approach or insertion reference, saturated at the velocity limit; the warm-started sequence is blended toward this reference with weight \(0.6\). We then compute the optimal control sequence via the MPPI update rule \cite{williams2017mppi}
\vspace{-4pt}
\begin{equation}\label{eq:mppiupdate}
\begin{array}{c}
\mathbf{u}^\star \leftarrow \hat{\mathbf{u}} + \sum_{j=1}^{K} \rho^{(j)}\, \boldsymbol{\epsilon}^{(j)},
\end{array}
\end{equation}
and shift \(\mathbf{u}^\star\) by one step to obtain \(\hat{\mathbf{u}}\) for the next solve, appending a zero terminal input. This warm start preserves useful structure across successive MPC problems, while every candidate is re-evaluated from the current measured state. Because finite \(K\) and \(N_p\) and the penalty in \eqref{eq:score} approximate \Cref{prb:rsbel}, the sampled MPPI controller does not guarantee satisfaction of the hard constraint imposed by the exact problem.

\subsection[Online Solution of Problem 1]{Online Solution of \Cref{prb:rsbel}}\label{ssec:opt}
At each time step \(t\), Algorithm~\ref{alg:belmpc} takes the prior belief \(b_{t-1}(\theta)\) and observation \(z_t\), updates to the posterior \(b_t(\theta)\) via \eqref{eq:partbelupd}, and solves \Cref{prb:rsbel} online. Only the first element \(u_t^\star\) of the resulting sequence \(\mathbf{u}^\star\) is applied, while the remainder warm-starts the next MPPI iteration at \(t+1\).

\section{Theoretical Guarantees}\label{sec:theory}
We now establish three properties of the controller defined by~\eqref{eq:obj}--\eqref{eq:cvarcnst}: a probabilistic safety guarantee that follows from the CVaR constraint~(\Cref{thm:safety}), a consistency result showing that the formulation recovers risk-neutral behavior in the appropriate limit~(\Cref{thm:limit}), and a cumulative safety bound for receding-horizon execution~(\Cref{thm:cumulative})
\begin{assumption}\label{asm:jhfinmom}
The random trajectory cost \(J_H(\mathbf{u})\) has finite first moment for every \(\mathbf{u} \in \mathcal{U}^H\).
\end{assumption}
\begin{assumption}\label{asm:uhcomp}
The admissible control set \(\mathcal{U}^H\) is compact, and the mappings
\(\mathbf{u} \mapsto \E_{b_t}[J_H(\mathbf{u})]\), \(\mathbf{u} \mapsto \CVaR_{\beta_c}(J_H(\mathbf{u}))\), and \(\mathbf{u} \mapsto \CVaR_{\beta_s}(-M_H(\mathbf{u}))\) are continuous on \(\mathcal{U}^H\).\end{assumption}
\Cref{asm:jhfinmom} holds for bounded or polynomially-growing costs, including standard quadratic functionals like~\eqref{eq:lossfcn}~\cite{shapiro2014lectures}. \Cref{asm:uhcomp} holds since joint velocity commands are saturated to a bounded set~\cite{williams2017mppi}, and continuity of \(\E_{b_t}\) and \(\CVaR_{\beta_s}\) follows from continuity of \(f\) and \(g\) in \(\mathbf{u}\)~\cite{rockafellar2000cvar,thrun2005probabilistic}. Thus, neither assumption is restrictive for the problem class considered here.

\begin{theorem}[CVaR Safety Implication]\label{thm:safety}
Let \(\beta_s \in (0,1)\) and suppose
\(\CVaR_{\beta_s}\!\bigl(-M_H\bigr) \le 0\). Then
\begin{equation}\label{eq:probsafe}
\begin{array}{l}
\Prob\bigl(M_H \ge 0\bigr) \ge \beta_s.
\end{array}
\end{equation}
\end{theorem}
\begin{proof}
Let \(L{:=}{-}{M_H}\). By the Rockafellar--Uryasev CVaR representation~\cite{rockafellar2000cvar} (see \cref{eq:vardef,eq:cvardef}),
\[
\VaR_{\beta_s}(L) \;\le\; \CVaR_{\beta_s}(L) \;\le\; 0,
\]
where the first inequality holds since \(\CVaR_{\beta_s}(L){=}\inf_{\eta\in\mathbb{R}}\bigl\{\eta + \tfrac{1}{1-\beta_s}\E[(L-\eta)^+]\bigr\} \ge \VaR_{\beta_s}(L)\), and the second is the constraint~\eqref{eq:cvarcnst}. Hence \(\VaR_{\beta_s}(L)\le 0\), which by \Cref{def:var} implies \(\Prob(L\le 0)\ge\beta_s\). Since \(L{:=}{-}{M_H}\), the events \(\{L\le 0\}=\{{-}{M_H}\le 0\}=\{M_H\ge 0\}\) are identical, so \(\Prob(L\le 0)\ge\beta_s\) yields~\eqref{eq:probsafe}. 
\end{proof}

\begin{remark}
Let \(V{:=}({-}{M_H})^+{=}\max\{{-}{M_H},0\}\), which is nonnegative and represents the magnitude of the safety violation. By Markov's inequality \cite{durrett2019probability}, for any \(\varepsilon\,{>}\,0\),
\begin{equation}
\Prob({-}{M_H} \ge \varepsilon) = \Prob(V \ge \varepsilon) \le \frac{\E[V]}{\varepsilon}.
\end{equation}
Thus, a mean-based bound controls safety only through the expected violation magnitude and can be loose when rare but severe violations dominate the tail. By contrast, the guarantee in \eqref{eq:probsafe} follows from the CVaR constraint \(\CVaR_{\beta_s}\left({-}{M_H}\le 0\right)\), which directly constrains the adverse tail of the safety-margin distribution. Hence \eqref{eq:probsafe} is stronger than what one obtains from Markov's inequality applied only to the expected violation.
\end{remark}

\begin{theorem}[Risk-Neutral Limit]\label{thm:limit}
Under Assumptions~\ref{asm:jhfinmom} and \ref{asm:uhcomp}, let \(\mathbf{u}^\star(\lambda_r)\) denote a minimizer of~\eqref{eq:obj} subject to~\eqref{eq:cvarcnst}.
Then every cluster point of \(\{\mathbf{u}^\star(\lambda_r)\}\) as \(\lambda_r\;{\downarrow}\;0\), with \(\lambda_r\) approaching zero from above, minimizes the risk-neutral objective \(\E_{b_t}[J_H]\) over the same feasible set.
\end{theorem}
\begin{proof}
Denote the feasible set \(\mathcal{U}^H(\beta_s)\), defined as
\begin{equation}
\mathcal{U}^H(\beta_s) = \{\mathbf{u} \in \mathcal{U}^H \colon \CVaR_{\beta_s}({-}M_H(\mathbf{u})) \le 0\},
\end{equation}
which is independent of \(\lambda_r\). Since \(\mathbf{u}\mapsto\CVaR_{\beta_s}({-}M_H(\mathbf{u}))\) is continuous by~\Cref{asm:uhcomp}, \(\mathcal{U}^H(\beta_s)\) is a closed subset of the compact set \(\mathcal{U}^H\), hence compact.
Write the objective \eqref{eq:obj} as \(\varphi(\mathbf{u}, \lambda_r)\,{=}\,\E_{b_t}[J_H(\mathbf{u})] + \lambda_r\,\CVaR_{\beta_c}(J_H(\mathbf{u}))\).
Under \Cref{asm:uhcomp}, \(J_H(\mathbf{u})\) depends continuously on \(\mathbf{u}\), so \(\CVaR_{\beta_c}(J_H(\cdot))\) is continuous on \(\mathcal{U}^H(\beta_s)\).
Since \(\mathcal{U}^H(\beta_s)\) is compact, \(\CVaR_{\beta_c}(J_H(\cdot))\) attains a finite supremum \(C < \infty\) on \(\mathcal{U}^H(\beta_s)\), and therefore
\begin{equation}
\hspace{-7.5pt}\sup_{\mathbf{u} \in \mathcal{U}^H(\beta_s)} \bigl|\lambda_r\,\CVaR_{\beta_c}(J_H(\mathbf{u}))\bigr| \;{\le}\; {\lambda_r} C \;\to\; 0 \ \text{as } \lambda_r\;{\downarrow}\;0.
\end{equation}
Hence \(\varphi(\cdot,\lambda_r) \to \E_{b_t}[J_H(\cdot)]\) uniformly on \(\mathcal{U}^H(\beta_s)\).
Since \(\mathcal{U}^H(\beta_s)\) is compact, every sequence \(\{\mathbf{u}^\star(\lambda_r)\}\) has a cluster point \(\bar{\mathbf{u}} \in \mathcal{U}^H(\beta_s)\). For any such cluster point and any \(\mathbf{v} \in \mathcal{U}^H(\beta_s)\), the minimality of \(\mathbf{u}^\star(\lambda_r)\) gives \(\varphi(\mathbf{u}^\star(\lambda_r),\lambda_r) \le \varphi(\mathbf{v},\lambda_r)\). Taking the limit along the subsequence realizing \(\mathbf{u}^\star(\lambda_r)\to\bar{\mathbf{u}}\), and by uniform convergence~\cite{durrett2019probability}, we get
\begin{align}
\E_{b_t}[J_H(\bar{\mathbf{u}})]
&= \lim_{\lambda_r\;{\downarrow}\;0} \varphi(\mathbf{u}^\star(\lambda_r),\lambda_r)\nonumber\\
&\;\le\; \lim_{\lambda_r\;{\downarrow}\;0} \varphi(\mathbf{v},\lambda_r)
= \E_{b_t}[J_H(\mathbf{v})].
\end{align}
Since we chose \(\mathbf{v}\in\mathcal{U}^H(\beta_s)\) arbitrarily, it follows that \(\bar{\mathbf{u}} \in \arg\min_{\mathbf{u} \in \mathcal{U}^H(\beta_s)} \E_{b_t}[J_H(\mathbf{u})]\).
\end{proof}

\begin{theorem}[Cumulative Receding-Horizon Safety]\label{thm:cumulative}
Suppose the controller re-solves~\eqref{eq:obj}--\eqref{eq:cvarcnst} at $T$ successive time steps with horizon $H$ and safety level $\beta_s\in(0,1)$, satisfying~\eqref{eq:cvarcnst} at every re-solve. Let $M^{(j)}_H$ denote the safety margin at the $j$-th solve. Then
\vspace{-4pt}
\begin{equation}
\begin{array}{l}\label{eq:cumsafe}
\Prob\!\Bigl(\bigcap_{j=1}^{T}\bigl\{M^{(j)}_H\ge 0\bigr\}\Bigr)\;\ge\;1-T(1-\beta_s).
\end{array}
\end{equation}
The intersection in \eqref{eq:cumsafe} is the event that every predicted horizon at the \(T\) successive solves has a nonnegative safety margin.
\end{theorem}
\begin{proof}
By \Cref{thm:safety}, $\Prob(M^{(j)}_H\,{<}\,0)\le 1-\beta_s$ for each $j$. A union bound \cite{durrett2019probability} gives
\[
\Prob\!\Bigl(\bigcup_{j=1}^T\{M^{(j)}_H\,{<}\,0\}\Bigr)\le T(1-\beta_s),
\]
and taking the complement yields~\eqref{eq:cumsafe}.
\end{proof}
\begin{remark}
The bound is non-trivial while $T<1/(1-\beta_s)$; for $\beta_s=0.90$ this permits up to $T=9$ solves before it degenerates. For larger $T$, one can tighten $\beta_s$ adaptively or apply Boole--Fr\'echet corrections~\cite{hunterUpperBoundProbability1976} when violations are positively correlated. Because successive horizons overlap, the union bound does not exploit dependence among their safety events and can therefore be conservative.
\end{remark}

\section{Case Study: Vision-Guided Dexterous Stowing}\label{sec:casest}
We instantiate our framework on a simulated stowing task, in which a manipulator inserts a grasped box into a narrow bookcase slot with uncertain geometry (\Cref{fig:tspart}). The slot imposes clearance and contact force constraints that the controller must respect, as detailed below.
\subsection{Object Stowing Task and Simulation Environment}\label{ssec:task}
\begin{figure}
    \centering
    \fbox{\includegraphics[width=0.8\linewidth]{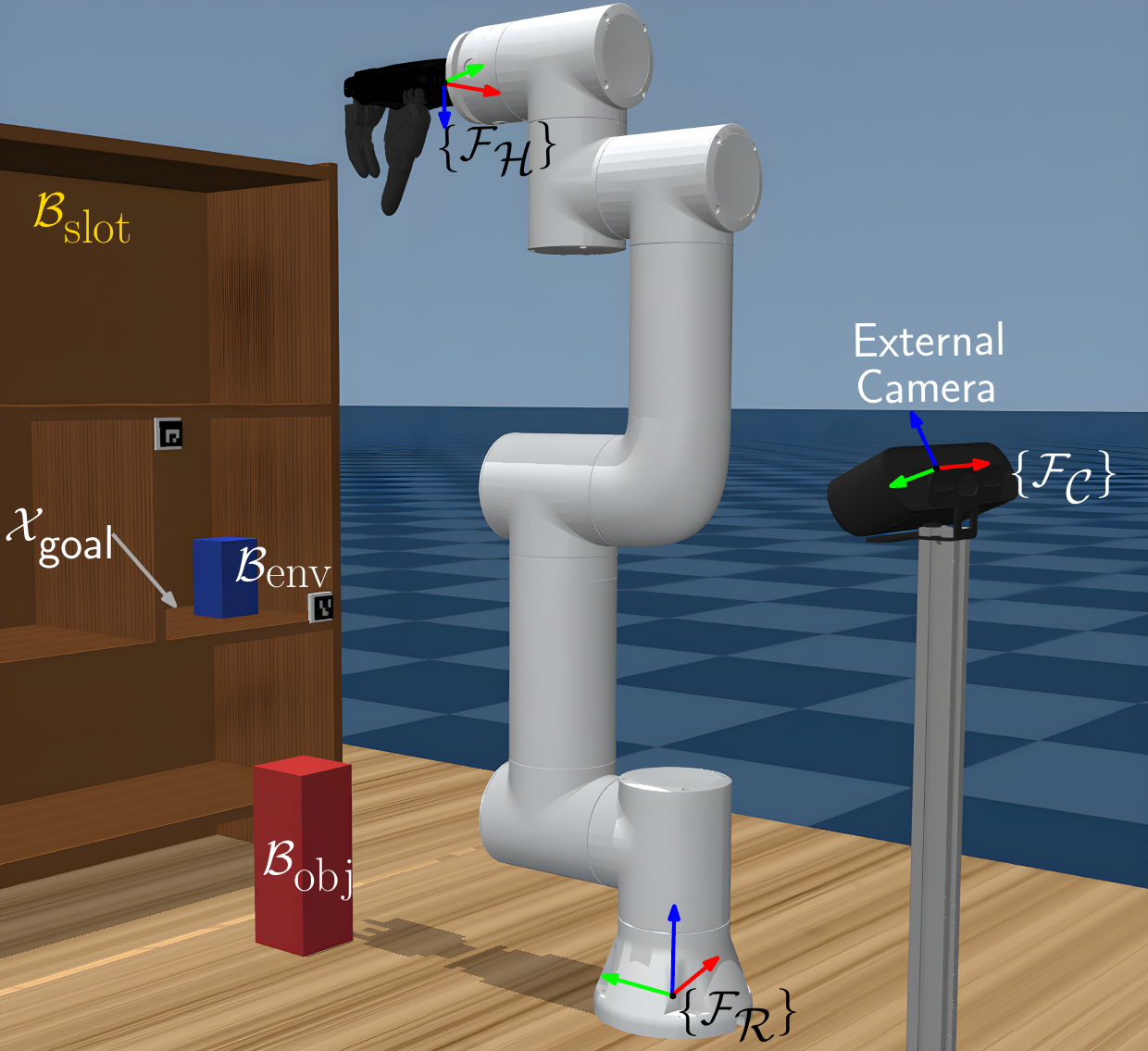}}
    \caption{
    \textbf{Object Stowing Task.}
    $\mathcal{F}_{\mathcal{R}}$, $\mathcal{F}_{\mathcal{H}}$, and $\mathcal{F}_{\mathcal{C}}$ denote the robot base, hand, and camera frames. $\mathcal{B}_{\mathrm{obj}}$ is the object to be transported, $\mathcal{B}_{\mathrm{slot}}$ denotes the receptacle, and $\mathcal{B}_{\mathrm{env}}$ represents previously stowed objects inside the receptacle. The goal region $\mathcal{X}_{\mathrm{goal}}$ specifies admissible placement poses within the receptacle, and all frames are expressed with respect to \(\mathcal{F}_{\mathcal{R}}\).
    }
    \label{fig:tspart}
\end{figure}
We use MuJoCo in our simulation study with a six-axis ZArm and an 11-joint RealHand L6 model. The transported box measures \(62.5\times62.5\times162.5\,\mathrm{mm}\), has mass \(0.15\,\mathrm{kg}\), and the receptacle slot is an \(80\times188\times180\,\mathrm{mm}\) region with only \(8.75\,\mathrm{mm}\) nominal lateral clearance per side. All quantities are expressed in the robot base frame $\mathcal{F}_{\mathcal{R}}$. The environment consists of an ArUco-marked receptacle $\mathcal{B}_{\mathrm{slot}}$, a transported object $\mathcal{B}_{\mathrm{obj}}$, and one or more previously stowed objects $\mathcal{B}_{\mathrm{env}}$ contained within the receptacle. The objective is to transport $\mathcal{B}_{\mathrm{obj}}$ into a goal region $\mathcal{X}_{\mathrm{goal}} \subset SE(3)$ within the receptacle. To avoid grasp planning complexity, we assume that a sufficiently stable grasp of $\mathcal{B}_{\mathrm{obj}}$ has already been established\footnote{\noindent We define a stable grasp as a tuple consisting of the robotic hand’s joint configuration and a set of object--hand contact points such that the resulting contact forces can balance any external wrench applied to the object, while also compensating for gravity.}. Let \(T_{\mathcal{R}}^{\mathcal{H}}(t) \in SE(3)\) denote the hand pose at time \(t\), and let \(T_{\mathcal G} \in SE(3)\) denote the fixed grasp transform from the hand frame to the object frame. The object pose at time \(t\) is then
\begin{equation}\label{eq:objpose}
x_t^{\mathrm{obj}}
\;:=\; T_{{\mathcal{R}}}^{\mathcal{B}_{\mathrm{obj}}}(t)
\;=\;T_{{\mathcal{R}}}^{{\mathcal{H}}}(t)\,T_{\mathcal{G}}.
\end{equation}
Let \(x_t^{\mathrm{obj}}=(p_t^{\mathrm{obj}},R_t^{\mathrm{obj}})\), with \(p_t^{\mathrm{obj}}\in\mathbb{R}^3\) and \(R_t^{\mathrm{obj}}\in SO(3)\). The object reaches \(\mathcal{X}_{\mathrm{goal}}\) when \(
\|p_t^{\mathrm{obj}}-p^\star\| \le \varepsilon_p \quad\text{and}\quad d_R(R_t^{\mathrm{obj}},R^\star) \le \varepsilon_R,\) where \((p^\star, R^\star) \in \mathbb{R}^3 \times SO(3)\) is the desired placement pose, known in advance and constructed from the perceived slot center and orientation, \(d_R(R_t^{\mathrm{obj}},R^\star):=\arccos{(\mathrm{Tr}(R_t^{\mathrm{obj})^{\top}}\cdot R^\star) - 1) /2)}\) is a rotation distance metric (geodesic distance on \(SO(3)\)), and \(\varepsilon_p > 0\) and \(\varepsilon_R\,{>}\,0\) are position and orientation tolerances.

\subsection{Instantiation of the General Framework}\label{ssec:inst}
We now cast the stowing task of \Cref{ssec:task} as a concrete instantiation of the risk-sensitive belief-space control framework introduced in \Cref{sec:prbform}, making explicit the correspondence between task-specific quantities and the general formulation.
\begin{enumerate}[labelindent=\parindent, leftmargin=*]
\item[i.] \emph{Latent Parameter}:
In the case study, $\theta$ represents the true slot geometry, which we model as
\(
\theta = \big(T_{\mathrm{slot}},\, w_{\mathrm{slot}}\big),
\)
where \(T_{\mathrm{slot}} \in SE(3)\) is the slot frame pose and \(w_{\mathrm{slot}} \in \mathbb{R}_+\) is the slot half-width, both inferred from noisy camera observations. We model this uncertainty by setting the initial belief over the translational component \(p_{\mathrm{slot}} \in \mathbb{R}^3\) of \(T_{\mathrm{slot}}\) to 
\begin{equation}
p_{\mathrm{slot}} \sim \mathcal{N}(\hat{p}_{\mathrm{slot}},\, \sigma_p^2 I_3), \quad \sigma_p = 12\,\mathrm{mm},
\end{equation}
where \(\hat{p}_{\mathrm{slot}}\) is the camera-based point estimate, and \(I_3\) is the \(3\times3\) identity matrix.

\item[ii.]\emph{Belief Update}:
We represent \(b_t(\theta)\) with a weighted particle filter over \(N_p\) samples. Particles are reweighted by the ArUco detection likelihood and resampled when the effective sample size drops below \(N_p/2\).

\item[iii.] \emph{State and Dynamics}:
The state \(x_t\) consists of the arm--hand configuration, its velocity, and the transported object's pose \(x_t^{\mathrm{obj}}=(p_t^{\mathrm{obj}},R_t^{\mathrm{obj}})\), where \(x_t^{\mathrm{obj}}\) is induced by the fixed grasp transform in \eqref{eq:objpose}. The control input \(u_t \in \mathcal{U} \subseteq \mathbb{R}^m\) is the joint velocity command of the arm--hand system, where \(m\) is the number of controllable joints. MuJoCo integrates the dynamics every 2\,ms, and each discrete control step in \eqref{eq:dyn} comprises 25 integration steps, or 50\,ms. Since the grasp is assumed fixed, the dynamics model is used to predict transport, insertion, and interaction loads rather than to synthesize new hand--object contacts.

\item[iv.] \emph{Cost Function}:
We use the stage cost
\begin{equation}\label{eq:lossfcn}
\ell(x_t, u_t, \theta) = \|p_t^{\mathrm{obj}} - p_{\mathrm{ref}}(\theta)\|_{\bar{Q}}^2 + \|u_t\|_{\bar{R}}^2,
\end{equation}
where \(p_t^{\mathrm{obj}}\) is the \(t\)-step transported object position, \(p_{\mathrm{ref}}(\theta)\) is the approach or insertion waypoint obtained from the slot center associated with \(\theta\), \(\bar{Q} \in \mathbb{R}^{3\times 3}\) is a positive semi-definite position tracking weight, and \(\bar{R} \in \mathbb{R}^{m\times m}\) is a positive definite control regularization weight over the joint velocity commands. Before insertion, the reference is offset outward along the slot axis; during insertion, it moves to the desired placement position. The terminal cost \(\phi\) penalizes final placement error in position and orientation relative to \((p^\star, R^\star)\) to encourage efficient transport and accurate placement.

\item[v.] \emph{Controller Parameters}:
We solve~\eqref{eq:obj}--\eqref{eq:cvarcnst} with an MPPI implementation (Algorithm~\ref{alg:belmpc}) using \(K = 64\) candidate trajectories, horizon \(H = 12\), and \(N_p = 16\) belief particles per rollout. With the 50\,ms control step, the horizon spans 0.6\,s. See \Cref{tab:params} for a full parameter list.

\item[vi.] \emph{Safety Margin:}
We define the safety function as the minimum of three physically meaningful margins,
\begin{equation}\label{eq:clrnc}
h(x_t, \theta)
= \min\!\left\{
\begin{aligned}
&d_{\mathrm{env}}(x_t,\theta) - d_{\min},\\
&\bar f_{\mathrm{env}} - \|f_{\mathrm{env}}(x_t,\theta)\|,\\
&\bar f_{\mathrm{grasp}} - \|f_{\mathrm{grasp}}(x_t,\theta)\|
\end{aligned}
\right\},
\end{equation}
where \(d_{\mathrm{env}}(x_t,\theta)=\mathrm{dist}\!\big(x_t^{\mathrm{obj}},\, \mathcal{B}_{\mathrm{slot}}(\theta)\cup\mathcal{B}_{\mathrm{env}}\big)\) is the minimum signed distance between the transported object and the receptacle or previously stowed objects, \(d_{\min}>0\) is a clearance buffer, \(f_{\mathrm{env}}(x_t,\theta)\) is the resultant object--environment contact force, and \(f_{\mathrm{grasp}}(x_t,\theta)\) is a surrogate for the load transmitted through the grasp. The thresholds \(\bar f_{\mathrm{env}}\) and \(\bar f_{\mathrm{grasp}}\) bound environmental contact and grasp load, respectively. Hence, \(h(x_t,\theta)\ge 0\) enforces clearance, low-force contact, and grasp maintenance along the trajectory.

\begin{remark}[Rationale for Force-Based Safety Margin]
We use the force-based surrogate in \eqref{eq:clrnc} because the dominant failure modes are excessive contact forces and loss of clearance, both directly measurable and consistent with the simulation signals available to the controller. Since \(\sigma_p = 12\,\mathrm{mm}\) (see (i.)) exceeds our prescribed nominal lateral shelf clearance of \(8.75\,\mathrm{mm}\) per side, collision avoidance as well as the overall safe insertion task are inherently probabilistic, motivating the CVaR constraint in~\eqref{eq:cvarcnst}.
\end{remark}

\begin{table}[t]
\centering
\caption{Simulation parameters.}
\label{tab:params}
\setlength{\tabcolsep}{2pt}
\renewcommand{\arraystretch}{0.8}
\begin{tabular}{@{}llc@{}}
\toprule
Symbol & Description & Value \\
\midrule
\multicolumn{3}{l}{\textit{MPPI}} \\
\midrule
$K$         & Candidate trajectories       & $64$ \\
$H$         & Horizon steps                & $12$ \\
$N_p$       & Belief particles per rollout & $16$ \\
$\lambda$   & Temperature                  & $0.4$ \\
$\Sigma$    & Perturbation covariance      & $\sigma^2 I_3$;\; $\sigma\!\in\!\{0.08,0.015\}$ m/s\\
\midrule
\multicolumn{3}{l}{\textit{Risk}} \\
\midrule
$\beta_c$   & Performance CVaR level       & $\in \{0.5, 0.90, 0.95\}$ \\
$\beta_s$   & Safety CVaR level            & $\in \{0.5, 0.90, 0.95\}$ \\
$\lambda_r$ & Risk weight                  & $0.5$ \\
$\mu$       & Safety penalty weight        & $250$ \\
\midrule
\multicolumn{3}{l}{\textit{Task}} \\
\midrule
$d_{\min}$          & Clearance buffer             & $2\,\mathrm{mm}$ \\
$\sigma_p$          & Initial slot position\ std.\     & $12\,\mathrm{mm}$ \\
$\varepsilon_p$     & Stowing distance threshold   & $60\,\mathrm{mm}$ \\
$\varepsilon_R$     & Stowing orientation error tolerance   & $0.35\,$ radians \\
$\bar{f}_{\mathrm{env}}$ & Contact force limit     & $80\,\mathrm{N}$ \\
$\bar{f}_{\mathrm{grasp}}$ & Grasp load limit      & $80\,\mathrm{N}$ \\
\midrule
\multicolumn{3}{l}{\textit{Particle Filter}} \\
\midrule
$N_{\mathrm{thr}}$  & Resampling ESS threshold     & $N_p/2$ \\
\bottomrule
\end{tabular}
\end{table}
\end{enumerate}

\subsection{Chance-Constrained Baseline}\label{ssec:ccbl}
We compare our method against a point-estimate, risk-neutral chance-constrained MPPI baseline (\ccmppi) that uses the same dynamics model, stage cost, terminal cost, and rollout budget, but omits the CVaR term in \eqref{eq:obj} and evaluates \(M_H\) at the posterior mean \(\bar\theta = \mathbb{E}_{b_t}[\theta]\). At each time step, \ccmppi\ solves
\begin{align}
\min_{\mathbf{u} \in \mathcal{U}^H} \quad & J_H(\mathbf{u},\bar\theta,x_t) \label{eq:ccobj} \\
\text{subject to} \quad & M_H(\mathbf{u},\bar\theta) \ge -\delta_H, \label{eq:cccnst}
\end{align}
where \(\delta_H>0\) is a soft margin tolerance. Because \ccmppi\ uses only \(\bar\theta\), it does not represent the full belief and does not isolate the CVaR objective term from the CVaR safety constraint.

\subsection{CVaR Computation}\label{ssec:cvar}
The per-candidate scores in~\eqref{eq:score} require empirical estimates of the expectation and CVaR quantities in~\eqref{eq:empexp}, computed from the $N_p$ belief-particle rollouts of \Cref{ssec:alg}. For the performance CVaR, $\mathcal{I}_{\beta_c}^{(j)}$ indexes the $\lceil(1-\beta_c)N_p\rceil$ highest costs $\{J_H^{(i,j)}\}_{i=1}^{N_p}$; for the safety CVaR, $\mathcal{I}_{\beta_s}^{(j)}$ indexes the $\lceil(1-\beta_s)N_p\rceil$ largest values of $\{{-}{M_H^{(i,j)}}\}_{i=1}^{N_p}$, i.e., the most severely violated margins. The rollouts require $\mathcal{O}(K N_p H)$ dynamics evaluations, and both empirical CVaRs are computed in $\mathcal{O}(K N_p \log N_p)$ by sorting and averaging, which adds negligible overhead relative to the MuJoCo forward rollouts.
\section{Results \& Discussion}\label{sec:results}
We evaluate our control algorithm (Algorithm \ref{alg:belmpc}) in MuJoCo across 11 randomized trials per configuration, varying the initial belief spread and true slot pose offset. We compare three configurations: (i)~\emph{low risk aversion} (\(\beta_s = 0.50\)), (ii)~\emph{moderate risk aversion} (\(\beta_s = 0.90\)), and (iii)~\emph{high risk aversion} (\(\beta_s = 0.95\)). We set all other parameters to the defaults enumerated in \Cref{tab:params}.

\subsection{Primary Results}\label{ssec:primary}
\Cref{fig:stowing_sequences} contrasts successful placement at \(\beta_s=0.95\) with collision and retraction at \(\beta_s=0.50\).

\begin{figure}[t]
    \centering
    \includegraphics[width=\linewidth]{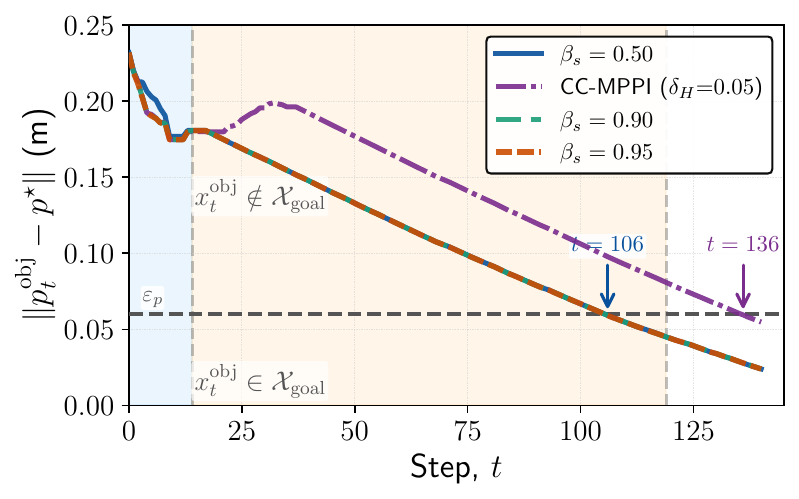}
    \caption{\textbf{Representative Trajectory.} Colored bands indicate task phases (\legendbox{phaseblue}~approach and \legendbox{phaseorange}~insert), and the dotted vertical line marks the transition between them. \(\varepsilon_{p}\) is the stowing proximity threshold. Both CVaR-MPPI (\(\beta_s \ge 0.90\)) configurations satisfy the terminal placement criteria, whereas \(\beta_s{=}0.50\) and \ccmppi\ (\(\delta_H{=}0.05\)) do not complete insertion.}
    \label{fig:traj}
\end{figure}

\begin{table}[t]
\centering
\caption{Insertion performance. \(\uparrow\) higher is better, \(\downarrow\) lower is better.}
\label{tab:results}
\setlength{\tabcolsep}{4pt}
\renewcommand{\arraystretch}{0.95}
\begin{tabular}{@{}lcccc@{}}
\toprule
& \(\beta_s{=}0.50\) & \(\beta_s{=}0.90\) & \(\beta_s{=}0.95\) & \ccmppi \\
\midrule
\multicolumn{5}{l}{\textit{Task Completion ($11$ trials)}} \\
\midrule
Success (\%) $\uparrow$           & $55$   & $73$   & $\mathbf{82}$ & $50$ \\
Contact rate (\%) $\downarrow$    & $27$   & $\mathbf{0}$   & $\mathbf{0}$   & $\mathbf{0}$ \\
\midrule
\multicolumn{5}{l}{\textit{Force \& Safety}} \\
\midrule
Mean ext.\ force (N) $\downarrow$ & $1.17$ & $\mathbf{0.00}$  & $\mathbf{0.00}$ & $0.35$ \\
Max ext.\ force (N) $\downarrow$  & $1118$ & $\mathbf{0}$   & $\mathbf{0}$  & $1$ \\
Min margin (mm) $\uparrow$        & $-28.4$ & $-23.4$ & $\mathbf{-23.4}$ & $-28.2$ \\
$\CVaR_{\beta_s}({-}{M_H})$ $\downarrow$ & $\mathbf{-0.151}$ & $-0.139$ & ${-0.135}$ & -- \\
\midrule
\multicolumn{5}{l}{\textit{Efficiency}} \\
\midrule
Final dist.\ (mm) $\downarrow$    & ${123.4}$ & $\mathbf{100.9}$ & $131.8$ & $121.2$ \\
Wall time (ms/step)               & $1169$  & $1172$  & $1193$  & $1207$ \\
\bottomrule
\end{tabular}
\end{table}
\Cref{tab:results} summarizes the quantitative results from $11$ randomized object insertion trials for varying performance and safety risk thresholds (\(\beta_c\,{=}\,\beta_s \in \{0.5, 0.9, 0.95\}\)) and fixed cost-risk weight \(\lambda_r{=}0.5\), contrasting them with the \ccmppi~baseline (\(\delta_H{=}0.05\)). \Cref{fig:traj} shows a representative object trajectory under each configuration. To reduce the effect of contact bounciness on the terminal object position, we apply a Savitzky-Golay filter~\cite{savitzky1964smoothing} with a window of 7 steps. 
\begin{figure}[b]
    \centering
    \includegraphics[width=\linewidth]{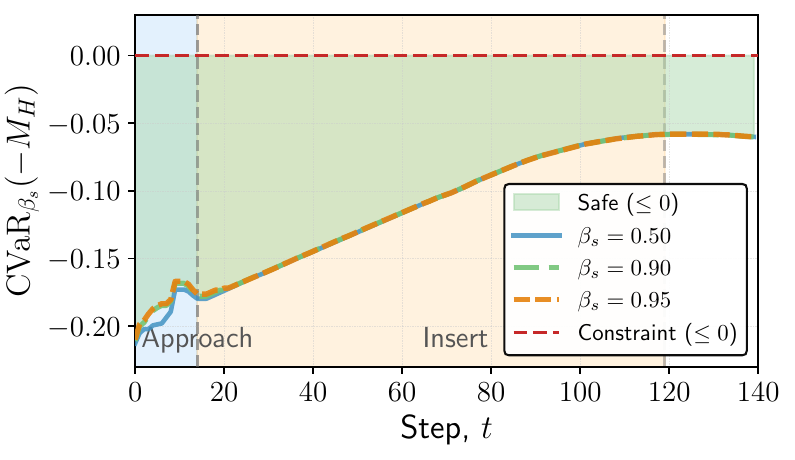}
    \caption{\textbf{CVaR Safety Constraint under Varying \(\boldsymbol{\beta_s}\).} Each panel shows the empirical \(\CVaR_{\beta_s}({-}{M_H})\) used by the MPPI score; green indicates a nonpositive estimate, and the red dashed line marks the constraint boundary. Higher \(\beta_s\) evaluates a narrower worst-case tail (worst 5\% vs.\ 50\%), moving the trace closer to the boundary. Only \(\beta_s \ge 0.90\) eliminates exterior contact~(\Cref{tab:results}).}
    \label{fig:cvar}
\end{figure}
\Cref{fig:cvar} shows the CVaR safety statistic for each configuration. Because \(\CVaR_{\beta_s}({-}{M_H})\) evaluates the worst \((1{-}\beta_s)\) tail, a higher \(\beta_s\) focuses on a narrower extreme, making the empirical constraint harder to satisfy and moving the trace closer to the boundary. The \(\beta_s{=}0.50\) trajectory has a nonpositive particle estimate but incurs a large simulated contact force, whereas \(\beta_s{=}0.95\) maintains wider physical clearances and eliminates exterior contact.

\subsection{Safety Confidence Level Ablation}\label{sssec:abl_beta}
To test the effect of the safety confidence level (\(\beta_s\)), we sweep \(\beta_s\) across three values (0.50, 0.9, and 0.95) recording success rate, insertion contact severity, and clearance. We report our results in \Cref{tab:abl_beta}.
\begin{table}[t]
\centering
\caption{Effect of safety confidence level \(\beta_s\) (initial position std.\ $\sigma_p{=}12$\,mm, exceeding the $8.75$\,mm lateral clearance). \(\uparrow\) higher is better, \(\downarrow\) lower is better.}
\label{tab:abl_beta}
\begin{tabular}{@{}lccc@{}}
\toprule
Metric & \(\beta_s{=}0.50\) & \(\beta_s{=}0.90\) & \(\beta_s{=}0.95\) \\
\midrule
Success rate (\%) $\uparrow$       & $55$   & $73$   & $\mathbf{82}$ \\
Contact rate (\%) $\downarrow$       & $27$   & $\mathbf{0}$   & $\mathbf{0}$   \\
Max ext.\ force (N) $\downarrow$     & $1118$ & $\mathbf{0}$   & $\mathbf{0}$ \\
\(\CVaR_{\beta_s}({-}{M_H})\) $\downarrow$ & $\mathbf{-0.151}$ & $-0.139$ & ${-0.135}$ \\
\bottomrule
\end{tabular}
\end{table}
At \(\beta_s{=}0.50\), the low-confidence controller has a \(27\%\) exterior-contact rate. Raising \(\beta_s\) to \(0.90\) or \(0.95\) eliminates exterior contact across the 11 trials, while simultaneously improving the success rate from \(55\%\) to \(82\%\). The reported \(1118\,\mathrm{N}\) value is the peak object--receptacle contact force computed by MuJoCo. However, because we did not calibrate the contact parameters against physical force measurements, this value indicates severe simulated contact but does not predict the corresponding real-world contact force.

\begin{figure*}[t]
    \centering
    \setlength{\fboxsep}{0pt}
    \setlength{\fboxrule}{0.5pt}
    \setlength{\tabcolsep}{2pt}
    \begin{tabular}{@{}cccc@{}}
    \multicolumn{4}{@{}l@{}}{\small (a) High Safety Confidence (\(\beta_s=0.95\))} \\[2pt]
    \fbox{\includegraphics[width=\dimexpr0.242\textwidth-2\fboxrule\relax]{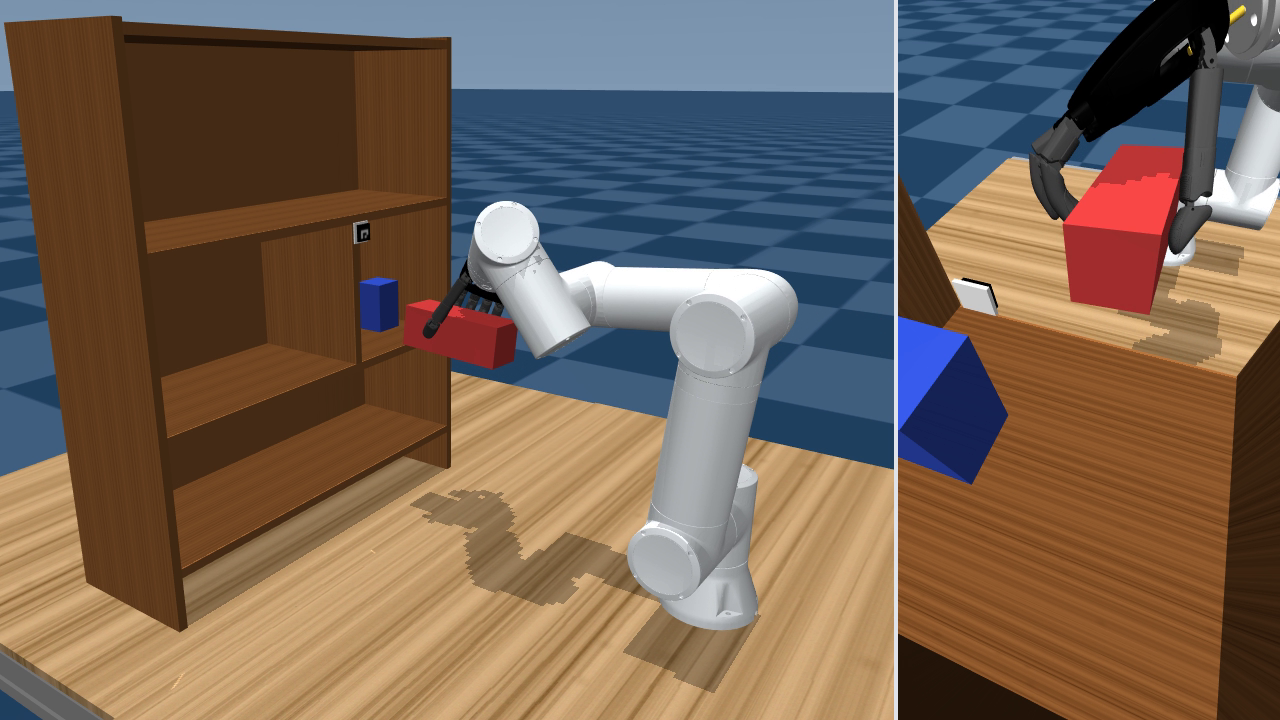}} & \fbox{\includegraphics[width=\dimexpr0.242\textwidth-2\fboxrule\relax]{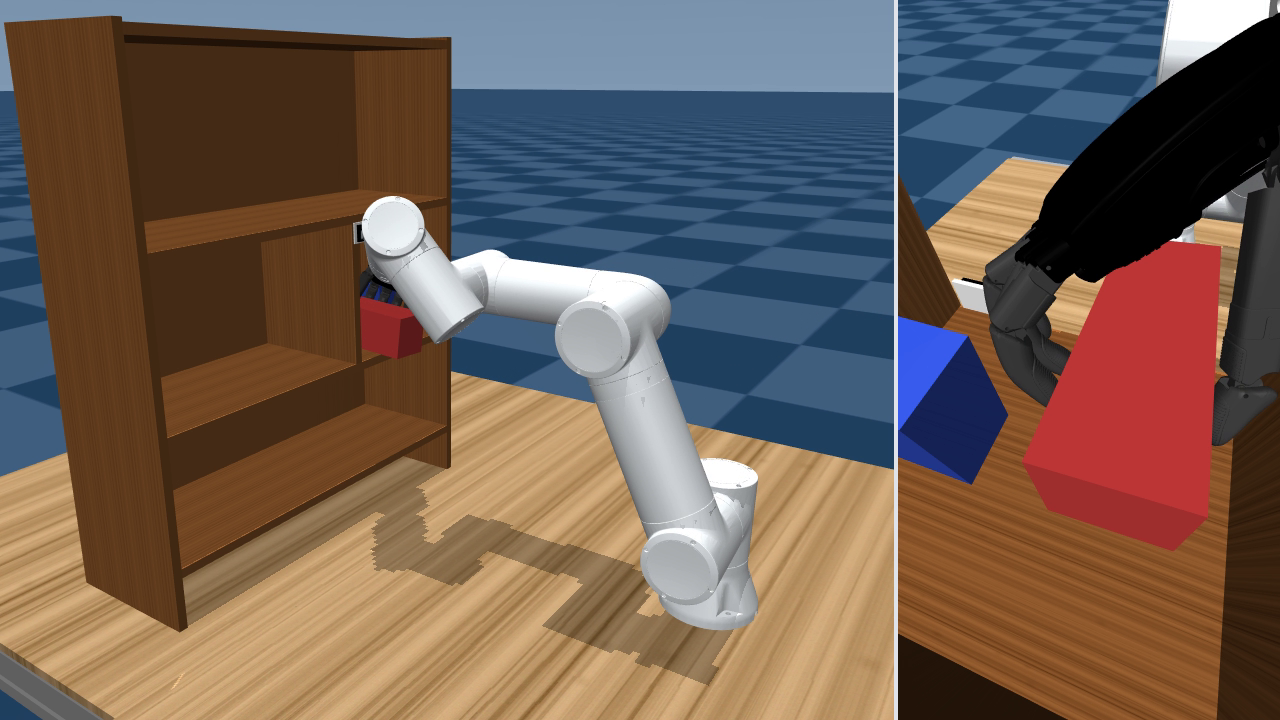}} & \fbox{\includegraphics[width=\dimexpr0.242\textwidth-2\fboxrule\relax]{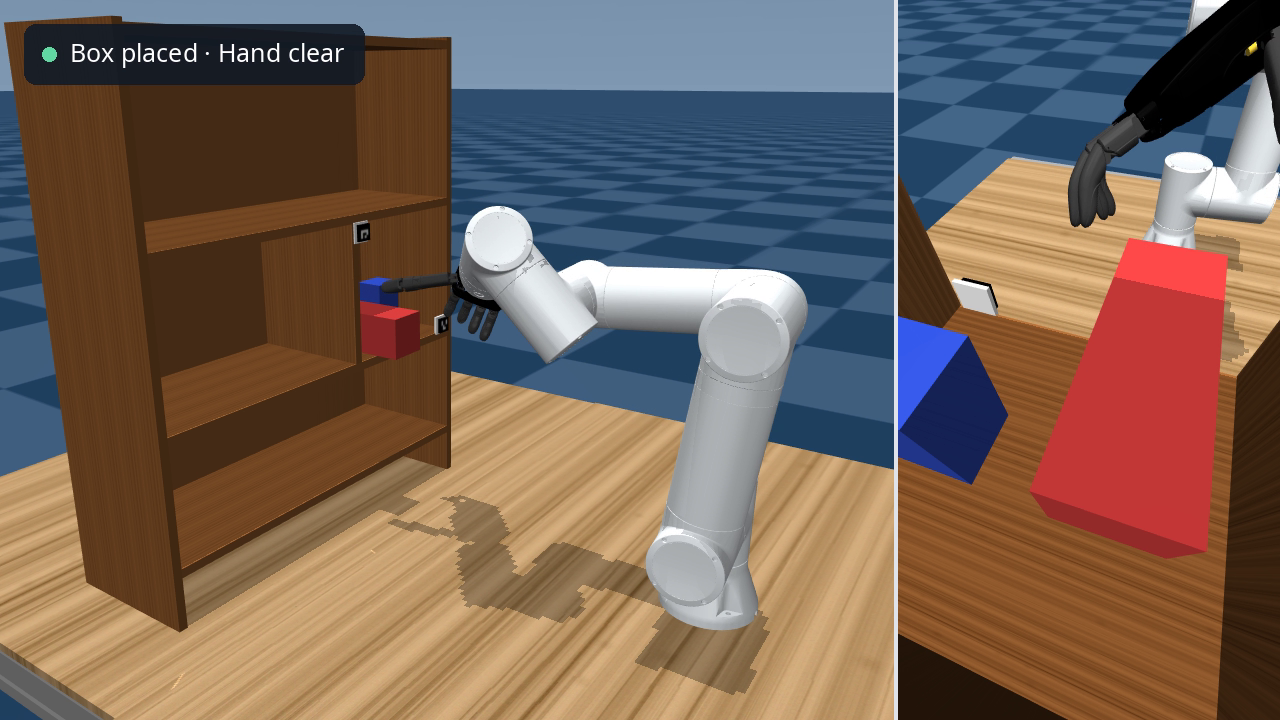}} & \fbox{\includegraphics[width=\dimexpr0.242\textwidth-2\fboxrule\relax]{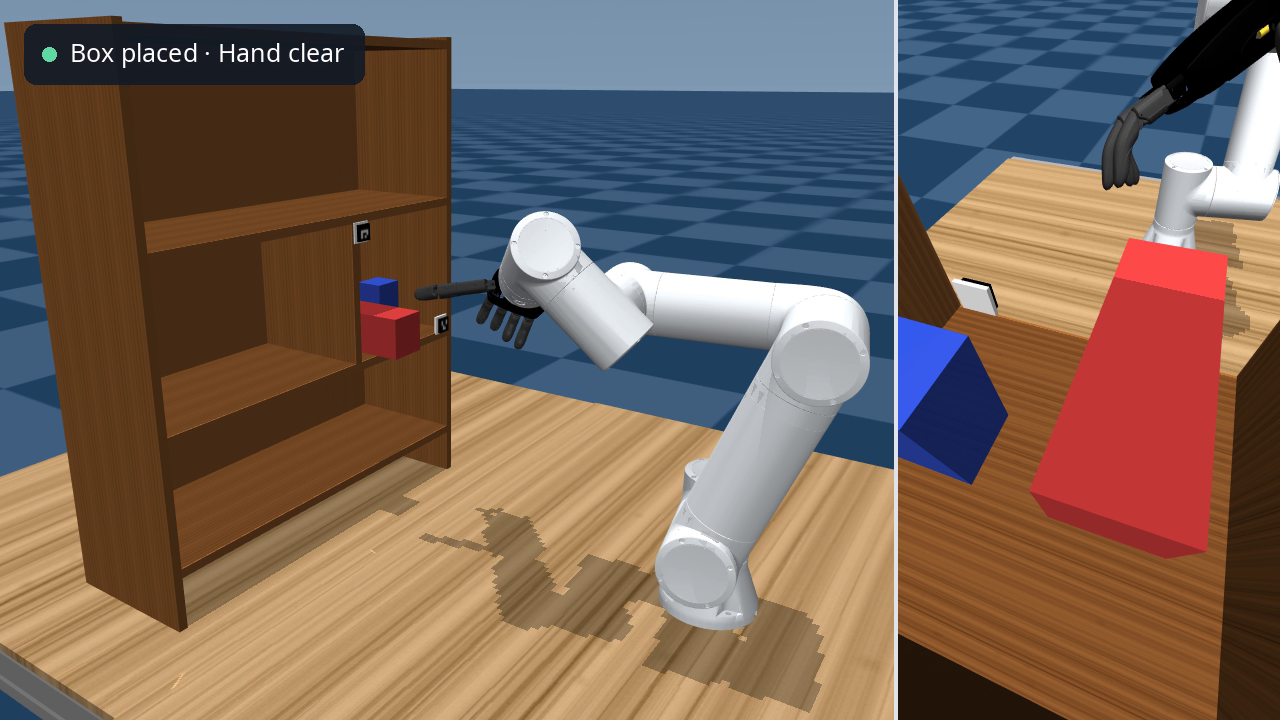}} \\
    {\scriptsize Approach (4.0 s)} & {\scriptsize Insert (9.6 s)} & {\scriptsize Place (13.2 s)} & {\scriptsize Retreat (15.0 s)} \\[5pt]
    \multicolumn{4}{@{}l@{}}{\small (b) Low Safety Confidence (\(\beta_s=0.50\))} \\[2pt]
    \fbox{\includegraphics[width=\dimexpr0.242\textwidth-2\fboxrule\relax]{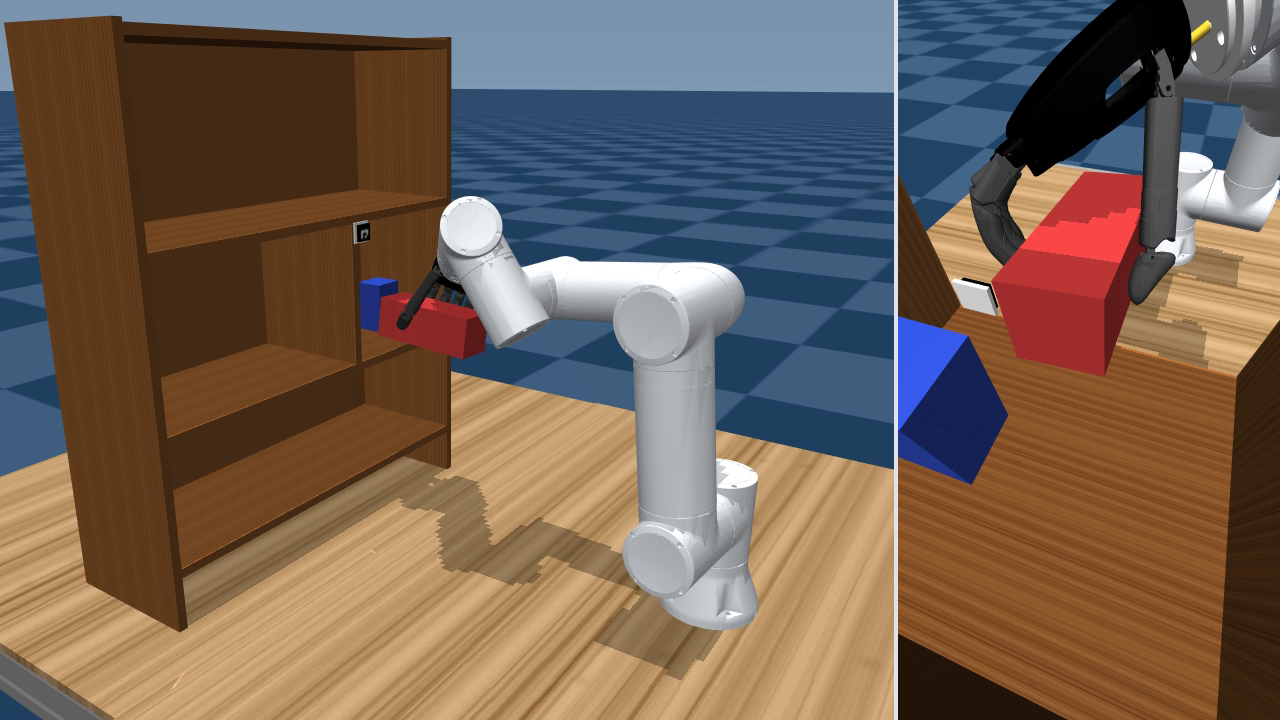}} & \fbox{\includegraphics[width=\dimexpr0.242\textwidth-2\fboxrule\relax]{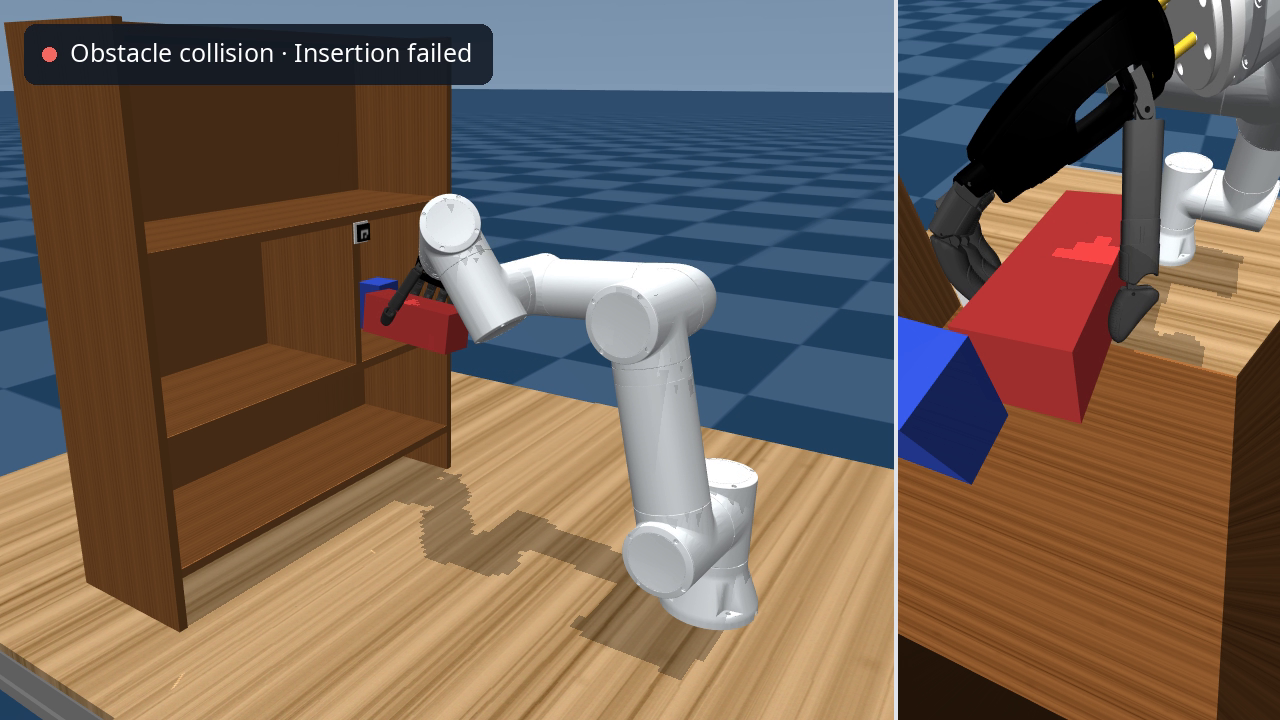}} & \fbox{\includegraphics[width=\dimexpr0.242\textwidth-2\fboxrule\relax]{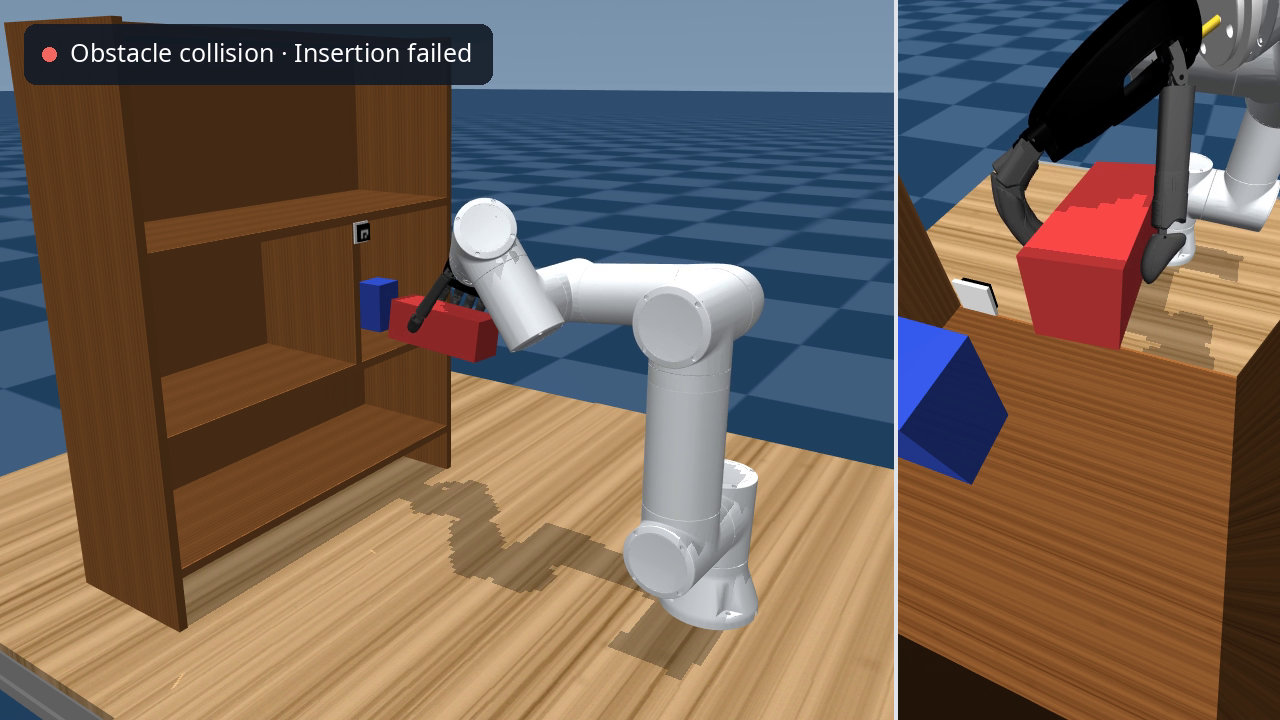}} & \fbox{\includegraphics[width=\dimexpr0.242\textwidth-2\fboxrule\relax]{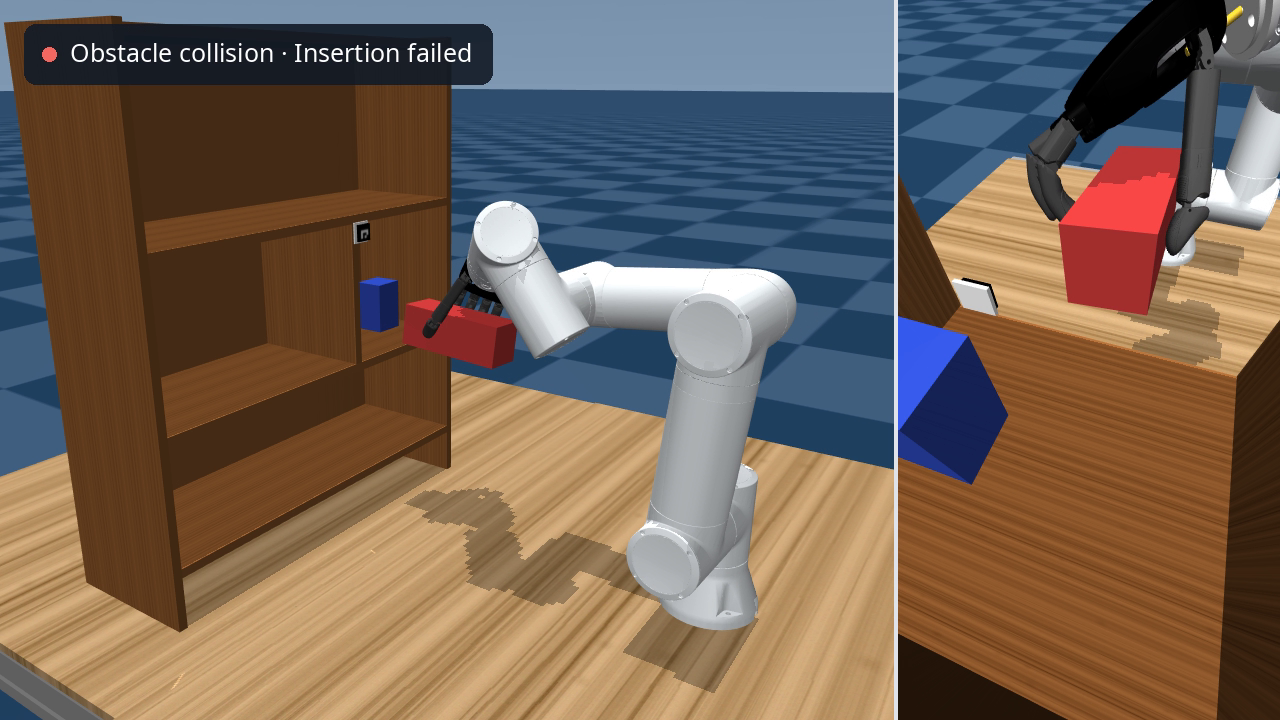}} \\
    {\scriptsize Insert (5.6 s)} & {\scriptsize Collision (6.2 s)} & {\scriptsize Retract (7.8 s)} & {\scriptsize Failure (9.4 s)} \\[5pt]
    \end{tabular}
    \caption{\textbf{Closed-Loop Stowing under Different Safety Confidence Levels.} The high-confidence execution (top) progresses through approach, insertion, placement, and retreat, while the low-confidence execution (bottom) collides during insertion and retracts without successful placement. Each row progresses from left to right. Times indicate elapsed task execution time.}
    \label{fig:stowing_sequences}
\end{figure*}
\subsection{Discussion}\label{ssec:discussion}
\paragraph{Safety-Performance Tradeoff}
\Cref{tab:results,tab:abl_beta} show that increasing \(\beta_s\) eliminates exterior contact (27\%\,to\,0\% contact rate) and improves success from 55\% to 82\%. Final placement error is not monotone in \(\beta_s\): the moderate-risk configuration has the smallest final distance, while the high-risk configuration has the highest success rate.

\paragraph{CVaR-MPPI vs.\ CCMPPI}
Both \ccmppi\ (\(\delta_H{=}0.050\)) and CVaR-MPPI (\(\beta_s{=}0.90\)) achieve zero exterior contact, yet CVaR-MPPI attains a higher success rate (\(73\%\) vs.\ \(50\%\)); at \(\beta_s{=}0.95\), the success rate is \(82\%\). Here, zero exterior contact does not itself imply successful placement: success separately requires the terminal position and orientation criteria. Because \ccmppi\ uses a point estimate and omits both CVaR terms, this comparison evaluates the complete controllers rather than isolating either CVaR component.

\paragraph{Computational Cost}
All configurations run at approximately \(1{,}200\)\,ms per control step on a single CPU core (\Cref{tab:results}), dominated by the \(K \times N_p\) forward rollouts in MuJoCo (\(1{,}024\) total). This is about 24 times the 50\,ms control period and is therefore not suitable for real-time industrial deployment without parallel rollout evaluation or model acceleration; the CVaR estimation overhead is negligible relative to simulation.

\paragraph{Limitations}
The evaluation is limited to one simulated task with 11 trials per configuration; it does not include hardware experiments or independent perturbations of the planning model relative to the execution model. It also assumes a rigid grasp, a simplified clearance model with a finite set of surface points, and independent observation noise. These assumptions bound the present empirical conclusions, while hardware, model-mismatch, and additional-task studies remain necessary to assess broader applicability.

\FloatBarrier

\section{Conclusion}\label{sec:conc}
We presented a belief-space MPC framework for safe control under latent uncertainty. By solving a risk-regularized finite-horizon optimal control problem with a CVaR safety constraint on the trajectory margin, the controller produces actions that optimize expected performance while bounding the probability of safety violations. The CVaR constraint implies a probabilistic safety guarantee (\Cref{thm:safety}), the formulation recovers risk-neutral control in the appropriate limit (\Cref{thm:limit}), and a union-bound argument extends the per-horizon guarantee to cumulative safety across receding-horizon re-solves (\Cref{thm:cumulative}). The vision-guided dexterous insertion study shows improved success and reduced contact at higher safety confidence levels, and the complete CVaR-MPPI controller outperforms the point-estimate CCMPPI baseline. Future work will address correlated observation noise and extension to time-varying latent parameters.

\renewcommand{\baselinestretch}{0.9}\selectfont
\makeatletter
\let\arxivoldbibliography\thebibliography
\renewcommand{\thebibliography}[1]{\arxivoldbibliography{#1}\fontsize{7}{8.5}\selectfont\setlength{\itemsep}{0pt}}
\makeatother
\bibliographystyle{ieeetr}
\bibliography{references}
\end{document}